%% file: dhh13.tex
\documentclass[11pt]{article}

\usepackage{enumerate}
\usepackage{pstricks,pstricks-add}
\usepackage{fullpage}
\usepackage{natbib}
\usepackage{amsmath,amssymb,amsthm}
\usepackage{subfigure}

\newtheorem{lemma}{Lemma}[section]
\newtheorem{definition}{Definition}[section]
\newtheorem{theorem}{Theorem}[section]
\newtheorem{proposition}{Proposition}[section]

\include{macro}

\newcommand{\STOC}[1]{}
\newcommand{\NOTSTOC}[1]{#1}

\begin{document}

\input{title}
\input{intro}
\input{prelim}

\input{benchmarks}

\input{clinching}

\input{clinchingVsWelfare}

\input{revenue-pe}
\input{revenue-sampling}
\input{revenue-approx}
\bibliographystyle{apalike}
\bibliography{dhh13}

\appendix

\input{appendix-clinching}
\input{no-budget}

\input{appendix-biasedsampling}

\end{document}

%% file: macro.tex
\newcommand{\given}{\mid}
\newcommand{\prob}[1]{\mathbf{Pr}\big[#1\big]}
\newcommand{\expect}[1]{\mathbf{E}\big[#1\big]}

\newcommand{\budget}{B}
\newcommand{\val}{v}
\newcommand{\vals}{{\mathbf \val}}
\newcommand{\valsmi}[1][i]{{\mathbf \val}_{-#1}}
\newcommand{\vali}[1][i]{{\val_{#1}}}
\newcommand{\alloc}{x}
\newcommand{\allocs}{{\mathbf\alloc}}
\newcommand{\alloci}[1][i]{{\alloc_{#1}}}

\newcommand{\payment}{p}
\newcommand{\payments}{{\mathbf\payment}}
\newcommand{\paymenti}[1][i]{\payment_{#1}}
\DeclareMathOperator{\EFO}{EFO}
\newcommand{\super}[1]{^{(#1)}}

\newcommand{\feasibles}{{\cal X}}
\newcommand{\mech}{{\cal M}}

\newcommand{\clock}{p}
\newcommand{\supply}{W}
\newcommand{\supplyi}[1][i]{\supply_{#1}}
\newcommand{\cumalloc}{X}
\newcommand{\cumalloci}[1][i]{\cumalloc_{#1}}
\newcommand{\cumfeas}{F}
\newcommand{\cumfeasi}[1][i]{\cumfeas_{#1}}
\newcommand{\budgeti}[1][i]{\budget_{#1}}

\newcommand{\pos}{w}
\newcommand{\posi}[1][i]{\pos_{#1}}
\newcommand{\poss}{(\posi[1],\ldots,\posi[n])}
\newcommand{\ipos}{\bar{\pos}}
\newcommand{\iposi}[1][i]{\ipos_{#1}}
\newcommand{\magent}{\kappa}

\newcommand{\marginal}{\delta}
\DeclareMathOperator{\Clinching}{Clinching}
\DeclareMathOperator{\EFOW}{EFO^W}

\newcommand{\coin}{p} 
\newcommand{\ruin}{r} 
\newcommand{\market}{M}
\newcommand{\sample}{S}

\newcommand{\mals}{\vals^{\market}}

\newcommand{\sals}{\vals^{\sample}}

\newcommand{\eval}{\tilde{\val}}
\newcommand{\evals}{\tilde{\mathbf \val}}

\newcommand{\evali}[1][i]{{\eval_{#1}}}
\newcommand{\ealloc}{\tilde{\alloc}}
\newcommand{\eallocs}{\tilde{\mathbf\alloc}}
\newcommand{\ealloci}[1][i]{{\ealloc_{#1}}}
\newcommand{\domind}{\eta}
\newcommand{\domi}[1][i]{\domind={#1}}
\newcommand{\nobudgetapprox}{7.5}
\newcommand{\budgetapprox}{10.0}
\DeclareMathOperator{\EFOR}{EFO^R}

\DeclareMathOperator{\PEoperator}{PE}
\newcommand{\PE}[1]{\PEoperator^{#1}}
\DeclareMathOperator{\BSPEoperator}{BSPE}
\newcommand{\BSPE}[1]{\BSPEoperator_{#1}}
\DeclareMathOperator{\PERoperator}{PER}
\newcommand{\PER}[1]{\PERoperator^{#1}}

\newcommand{\lag}{\lambda} 
\newcommand{\lvv}{\phi^\lag}

\newcommand{\lvvi}[1][i]{{\lvv_{#1}}}

\newcommand{\ilvvi}[1][i]{{\bar{\phi}_{#1}}}
\newcommand{\lrev}{R^\lag}
\newcommand{\ilrev}{\bar{R}^\lag}

%% file: title.tex
\begin{titlepage}
\title{\Large Prior-free Auctions for Budgeted Agents}
\author{
  Nikhil R. Devanur%
  \thanks{Microsoft Research, Redmond, WA.
  Email: {\texttt nikdev@microsoft.com}}\\
\and
  Bach Q. Ha%
  \thanks{Department of Electrical Engineering and Computer Science,
  Northwestern University, Evanston, IL.
  Email: {\texttt \{bach,hartline\}@u.northwestern.edu}.}\\
\and
  Jason D. Hartline\footnotemark[2]
}
\date{}
\maketitle

\input{abstract}

\thispagestyle{empty}

\end{titlepage}
\newpage
\setcounter{page}{1}

%% file: abstract.tex
\begin{abstract} 
We consider prior-free auctions for revenue and welfare maximization
when agents have a common budget.  The abstract environments we
consider are ones where there is a downward-closed and symmetric
feasibility constraint on the probabilities of service of the agents.
These environments include {\em position auctions} where slots with
decreasing click-through rates are auctioned to advertisers.  We
generalize and characterize the envy-free benchmark from \citet{HY-11}
to settings with budgets and characterize the optimal envy-free
outcomes for both welfare and revenue.  We give prior-free mechanisms
that approximate these benchmarks.  A building block in our mechanism
is a clinching auction for position auction environments.  This
auction is a generalization of the multi-unit clinching auction of
\citet{DLN-08} and a special case of the polyhedral clinching auction
of \citet{GMP-12}.  For welfare maximization, we show that this
clinching auction is a good approximation to the envy-free optimal
welfare for position auction environments.  For profit maximization,
we generalize the random sampling profit extraction auction from
\citet{FGHK-02} for digital goods to give a
$\budgetapprox$-approximation to the envy-free optimal revenue in
symmetric, downward-closed environments.  The profit maximization
question is of interest even without budgets and our mechanism is a
$\nobudgetapprox$-approximation which improving on the 30.4 bound of
\citet{HH-12}.
\end{abstract}

%% file: intro.tex
\section{Introduction}

%
%
Economic mechanisms that are less dependent on the assumptions of the
environment are more likely to be relevant \citep{W-87}.  The area of
{\em prior-free mechanism design} attempts to give mechanisms that
guarantee a good approximation to the designer's objective without
dependence on distributional assumptions on the agents' preferences.
The main open questions in prior-free mechanism design center around
the departure from the ideal single-dimensional and linear model of
preferences, i.e., where an agent's utility is given by her value for
service less the payment she is required to make.  A paradigmatic
example of a non-linear agent utility is one that is linear up to a
given budget that restricts the agent's maximum possible payment.  In this
paper we consider the designer's objectives of revenue and welfare
(separately) when agents have budgets, and we give simple prior-free
auctions that approximate a natural prior-free benchmark (for each
objective).

%
%
Mechanism design studies optimization on inputs that are the private
information of strategic agents who may misreport their information if
it benefits them.  Agents will not misreport only if they have no
incentive to, i.e., if their utilities are maximized by truthful
reporting.  The key challenge in designing mechanisms for strategic
agents, then, is that incentive constraints bind across possible agent
misreports.  A mechanism must therefore trade-off performance on one
input versus another.  For general objectives, e.g., welfare with
budgeted agents, profit, or makespan (for unrelated machines), there
is no single mechanism that is simultaneously optimal on all
inputs.\footnote{For these objectives, the designer's objective and
  the agents' objectives are fundamentally at odds and the incentive
  constraints of the agents do not permit the designer to obtain the
  same performance possible when the inputs are public.  These
  objectives contrast starkly to the objective of welfare maximization
  without budgets where there is no conflict in the designer's
  objective and the agents' objectives and the Vickrey-Clarke-Groves
  (VCG) mechanism is pointwise optimal \citep{Vic-61,Cla-71,Gro-73}.}
There are two approaches for addressing the non-pointwise-optimality
of mechanisms. The Bayesian approach, which is standard in economics,
assumes that the agents' preferences (inputs) are drawn from a known
distribution and the performance of the mechanism across different
inputs can be traded off so as to optimize its expected performance
with respect to this given distribution.  The Bayesian optimal mechanism,
therefore, depends on the distribution.  The prior-free approach,
which is currently being developed in computer science, instead looks
for a single mechanism that approximates an economically motivated
prior-free benchmark in worst-case over all inputs.

%

%
%
The first step in developing prior-free mechanisms is to identify an
appropriate prior-free benchmark.  \citet{HY-11} recently observed
that a simple and intuitive prior-free benchmark can be defined based
on a relaxation of the no-misreporting incentive constraint to a
no-envy constraint.  The advantage of the no-envy constraint is that
it binds pointwise on each input instead of across all inputs like
incentive constraints; therefore, there is always a pointwise optimal
envy-free outcome.  Furthermore, as \citet{HY-11} pointed out, often
this benchmark is an upper bound on the optimal performance on the
Bayesian optimal mechanism for any distribution; in these cases
approximating it pointwise gives a very strong performance guarantee.
Our first contribution is a generalization of the revenue-optimal
envy-free benchmark without budgets to the objectives of revenue and
welfare with budgets.

%
%
A mechanism must optimize its objective subject to incentive
constraints (discussed above), feasibility constraints (i.e.,
constraints on how agents can be served together), and budget
constraints.  It is most instructive to classify feasibility
constraints in terms of the sophistication required of constrained
optimization of a weighted sum of the set of agents served (or, for
randomized environments, probabilities of service).  An environment,
like that of digital good auctions, may be {\em unconstrained}.  An
{\em ordinal environment}, like those of position auctions (as
popularized by advertising on Internet search engines), is one where
the optimal algorithm is greedy on agents ordered by weight.  In a
general {\em cardinal environment}, like those of single-minded
combinatorial auctions, the weights of the agents are necessary for
optimization.  An environment is {\em symmetric} if the feasibility
constraint respects all permutations of the agent identities.  While
feasibility constraints limit which agents are served, budget
constraints limit the prices that agents pay.

Recent results of \citet{DLN-08} and \citet{GMP-12} have shown that a
generalization of the \citet{A-04} {\em clinching auction} is the only
{\em Pareto optimal} mechanism in ordinal environments.  At a
high-level, the clinching auction is given by an ascending price at
which each agent is allowed to claim any of the supply that would be
left over if that agent were given the last choice.  Pareto optimality
is the condition that there is no other feasible outcome where some
participant (including the designer) can be made strictly better off
without making some participant strictly worse off.  Pareto optimality
is a condition not an objective which means that it is not clear what
an approximation to Pareto optimality would be.  It is also not true
that all reasonable auctions must satisfy Pareto optimality.  The
Bayesian welfare-optimal auction for budgeted agents is not generally
Pareto optimal (see Sections~\ref{s:benchmarks} and~\ref{s:welfare});
and moreover, there does not generally exist Pareto-optimal auctions
for budgeted agents in cardinal environments \citep{GMP-12}.

%
%
Our first goal, given the limits of Pareto optimality, is to identify
a prior-free auction that approximates the envy-free optimal welfare
when agents have budgets.  The outcome of the clinching auction (for
ordinal environments) is envy free; however, it is not the
welfare-optimal envy-free outcome.  Moreover, given distribution over
agent values, the clinching auction is not the Bayesian optimal
auction for welfare either.  We give a simple closed form expression
for the clinching auction in symmetric ordinal environments with a
common budget and we show that it is a 2-approximation to the
envy-free benchmark.\footnote{While both \citet{DLN-08} and
  \citet{GMP-12} allow agents to have distinct budgets, the envy-free
  benchmark is only economically well motivated in symmetric
  environments therefore a common budget is required.  Our restriction
  to symmetric environments and in particular a common budget is
  reasonable as designing prior-free auctions for asymmetric
  environments is a challenge in itself even without budgets; for
  asymmetric environments only a few positive results are known, see,
  e.g., \citet{BBHM-08} and \citet{LR-12}, both of which are for
  unconstrained environments (e.g., digital goods).}


%
%
Our second goal is to identify a prior-free auction that approximates
the envy-free optimal revenue when agents have budgets.  For revenue
maximization without budgets \citet{HY-11} and \citet{HH-12} recently
gave general approaches for approximating the optimal envy-free
revenue.  The former extends a standard random sampling approach (for
digital good auctions) from \citet{GHW-01}; the latter extends an
approach based on ``consensus estimates'' and ``profit extraction''
from \citet{GH-03}.  Our approach is based on an extension of the {\em
  random sampling profit extraction} auction from
\citet{FGHK-02}.\footnote{Therefore all of the leading approaches for
  digital good auctions extend to more general environments. } Not
only is our mechanism the only one that is readily compatible with
budget constraints, but also the approximation factors we obtain,
relative to the revenue-optimal envy-free benchmark, are the best
known.  We show a $\budgetapprox$-approximation to the envy-free
optimal revenue in symmetric cardinal environments.  Moreover, without
budgets without budgets, our techniques give a
$\nobudgetapprox$-approximation which imptoves on the 30.4
approximation of \citet{HH-12}.

\paragraph{Summary of Results}

Our main conceptual contribution is the adaptation of the prior-free
mechanism design framework initiated by the literature on digital
goods (i.e., unconstrained environments), e.g., \citet{GHW-01} to the
structurally rich environments of \citet{HY-11} (including symmetric
ordinal environment and cardinal environments)\footnote{\citet{HY-11}
  refer to symmetric ordinal environments equivalently as {\em
    position environments} and {\em matroid permutation environments}
  and to symmetric cardinal environments as {\em downward-closed
    permutation environments}.} when agents' preferences are
non-linear as given by a common budget constraint.  Our technical
results are as follows:
\begin{itemize}
\item We give a characterization of the envy-free benchmark for
  welfare and revenue when agents have a common budget.  This
  characterization is via an extension of the characterization of
  Bayesian optimal auctions for agents with budgets of \citet{LR-96}
  to general distributions.\footnote{The \citet{LR-96}
    characterization holds for monotone hazard rate distributions with
    increasing density; under such assumptions, many of the novel
    properties of optimal auctions with budgets do not arise.}
\item We give a closed-form characterization of the clinching auction
  of \citet{GMP-12} in symmetric ordinal environments with a common
  budget.
\item We prove that the clinching auction is a 2-approximation to the
  envy-free optimal welfare in symmetric ordinal environments with a
  common budget.
\item We extend the random sampling profit extraction auction from
  \citet{FGHK-02} to symmetric cardinal environments with a common
  budget.  This generalization gives a $\budgetapprox$-approximation
  to the envy-free benchmark.  This is the first prior-free
  approximation of an economically well motivated benchmark for agents
  with budgets and it also improves (to $\nobudgetapprox$ from 30.4)
  on the best known prior-free approximation without budgets.
\item The clinching auction is not well defined in cardinal
  environments; the above auction converts the symmetric cardinal
  environment to symmetric ordinal environment where the clinching
  auction can be run and its objective is close to optimal (for the
  original cardinal environment).
\end{itemize}

\paragraph{Related Work}

%
%
The theory of Bayesian optimal auctions for welfare or revenue when
agents have budgets (a form of non-linear utility) is more complex
than that of revenue when agents have linear utility.  In the latter,
e.g., the revenue-optimal mechanism is given by optimizing {\em
  virtual values} which are given by a simple distribution-dependent
function of agents' values \citep{M-81}.  In the former, under some
restrictive distributional assumptions, a similar Lagrangian virtual
value approach gives the optimal mechanism (subject to careful choice
of the Lagrangian variable, see \citealp{LR-96}).

%
%
\citet{HY-11} defined the envy-free benchmark as a relaxation of the
Bayesian optimal auction that can be optimized pointwise.  Our
characterization of the envy-free benchmark for welfare and revenue
for agents with budgets combines and extends the results of
\citet{LR-96} and \citet{HY-11}.

%
%
There are three main techniques for designing revenue maximizing
prior-free auctions for digital goods (i.e., where there is no
feasibility constraint).  The {\em random sampling optimal price}
auction was defined by \citet{GHW-01}.  The {\em consensus estimate
  profit extraction} auction was defined by \citet{GH-03}.  The {\em
  random sampling profit extraction} auction was defined by
\citet{FGHK-02}.  The first two approaches were generalized to
symmetric cardinal environments by \citet{HY-11} and \citet{HH-12},
respectively.  We generalize the third approach to these environments.
Our generalization gives the best known approximation factor (to the
envy-free benchmark) without budgets (of \nobudgetapprox) and the
first approximation with budgets (of \budgetapprox).

%
%
Our mechanisms are based on the clinching auction of \citet{A-04}
generalized to multi-unit environments (a special case of ordinal
environments) with budgets by \citet{DLN-08} and ordinal environments
by \citet{GMP-12}.  There are two dimensions on which we can compare
our results to these prior studies of the clinching auction, (a)
whether agents have a common budget or distinct budgets and (b) the
feasibility constraint of the designer.  With respect to (a), our
results are weaker as we require a common budget, with respect to (b)
our symmetric ordinal environment is between multi-unit environments
and the general ordinal environments where the latter allows for
asymmetry.  The advantage of our restriction to environments that are
symmetric in budget (a) and feasibility (b), is that we are able to
derive a closed-form formula for the outcome of the clinching auction.
Finally, \citet{GMP-12} show that the clinching auction does not
generally extend to arbitrary cardinal environments; however, we show
that any symmetric cardinal environment contain a symmetric ordinal
environment for which the clinching auction performs well (with
respect to the objective on the original cardinal environment).
Moreover, we can effectively find this ordinal environment and run the
clinching auction on it without compromising the agent incentives.

\paragraph{Organization}

We give formal definitions of the model in Section~\ref{s:prelim}.  In
Section~\ref{s:benchmarks} we characterize the envy-free benchmarks
for agents with budgets.  In Section~\ref{s:welfare} we characterize
the clinching auction in position environments and show that it is a
2-approximation to the envy-free optimal welfare.  In
Section~\ref{s:revenue} we define the biased sampling profit
extraction auction and prove that it is a
$\budgetapprox$-approximation to the envy-free optimal revenue when
the agents have a common budget.  When the agents do not have a budget
constraint, the auction can be improved to a
$\nobudgetapprox$-approximation and this improvement is given in
Appendix~\ref{s:nobudget}.

%% file: prelim.tex
\section{Preliminaries}
\label{s:prelim}

\paragraph{Incentives}
We study auction problems for $n$ single-dimensional agents with a
common budget.  Each agent $i$ has a value $\vali$ for the service.  A
mechanism maps reported values $\vals = (\vali[1],\ldots,\vali[n])$ to
a probability that agent $i$ wins, $\alloci(\vals)$, and a payment
$\paymenti(\vals)$.\footnote{For clarity we will equate randomized
  mechanisms with deterministic mechanisms outputting fractional
  assignments and deterministic payments (both equal to their
  expectations).}  The agents are financially constrained by a budget
$\budget$ but otherwise are risk neutral.  Agent $i$'s utility from
the mechanism on reports $\vals$ is $\vali \alloci(\vals) -
\paymenti(\vals)$ if $\paymenti(\vals) \leq \budget$ and negative
infinity otherwise.

A mechanism is {\em budget respecting (BR)} if no agent pays more than
the budget on any valuation profile, i.e., for all $i$ and $\vals$,
$\paymenti(\vals) \leq \budget$.  A mechanism is {\em individually
  rational (IR)} if a risk-neutral agent weakly prefers to participate
in the mechanism than not.
$\forall i, \vals, \vali \alloci(\vals) - \paymenti (\vals) \geq 0.$
We say that a mechanism is {\em incentive compatible (IC)} if a
risk-neutral agent maximizes her utility by bidding her true value.
I.e., $\forall i, \vals, z,\  \vali \alloci(\vals) - \paymenti (\vals) \geq\vali \alloci(z, \valsmi) - \paymenti (z,\valsmi) .$
\cite{M-81} characterized incentive compatible mechanisms for single
dimensional agents (without budgets) as follows.
\begin{theorem}[\citealp{M-81}]\label{thm.myerson}
A mechanism is incentive compatible if and only if the allocation is
monotone non-decreasing in the reported values, i.e., for all $i$,
$\alloci(z,\valsmi)$ is monotone non-decreasing in $z$ and the
expected payments satisfy $\paymenti(z,\valsmi) = \vali
\alloci(z,\valsmi) - \int_0^{\vali} \alloci(z,\valsmi) dz.$
\end{theorem}

\paragraph{Feasibility}
As described above, an auction produces a randomized outcome for each
agent with probabilities denoted by $\allocs =
(\alloci[1],\ldots,\alloci[n])$.  We assume there is a feasibility
constraint that governs the set of such allocations that can be
produced.  We denote the space of feasible allocations by $\feasibles
\subset [0,1]^n$.  Our only requirement on this space is that it is
symmetric, convex, and downward-closed.\footnote{Our envy-free
  benchmark only makes sense in symmetric environments, mechanism
  design spaces are always convex if randomization is allowed, and
  downward closure says that if $\allocs \in \feasibles$ then
  $\allocs_{-i} \in \feasibles$ where $\allocs_{-i} =
  (\alloci[1],\ldots,\alloci[i-1],0,\alloci[i+1],\ldots,\alloci[n])$.}
Moreover, all we need from our feasibility constraint is that there is
an algorithm that (approximately) optimizes a linear sum of weights of
the agents served subject to it (and that any agent served can be
instead rejected); in these cases we instead view $\feasibles$ as the
induced allocation of the algorithm.

Given this algorithmic view, we partition the classes of feasibility
constraints by the kinds of algorithms that work.  If ``greedy by
weight'' is optimal then we refer to the feasibility constraint as
{\em ordinal} as only the order of the weights matters and not the
actual cardinal weights.  We refer to the more general case as {\em
  cardinal}.  Importantly the ordinal, symmetric case is identical to
the position auction environment under common study.  A {\em position
  environment} is given by a decreasing sequence of position weights
$\posi[1]\geq \cdots \geq \posi[n]$ and each agent can be matched to
at most one position.  The cardinal, symmetric case includes problems
considered in the literature such as the (symmetric restriction) of
the {\em polyhedral environments} of \citet{GMP-12} and the {\em
  downward-closed permutation environments} of \citet{HY-11}.



\paragraph{Envy-free Benchmarks}

The goal of prior-free mechanism design is to give a mechanism with
and a performance guarantee that holds point-wise, i.e., in worst
case, on valuation profiles.  Such a prior-free guarantee requires
comparison to a prior-free benchmark which is also defined point-wise
on valuation profiles.  Prior-free benchmarks that do not take into
account the incentive constraints of the mechanism design problem are
often inapproximable, but considering incentive constraints is
non-trivial because incentive constraints bind on possible agent
misreports and not point-wise on the valuation profile.  

\citet{HY-11} recently demonstrated that {\em envy-freedom} constrants
are a reasonable point-wise relaxation of incentive constraints.
Formally, an outcome $(\allocs, \payments)$ is envy free for valuation
profile $\vals$ if for all $i$ and $j$, agent $i$ does not prefer to
swap allocation and payment with agent $j$, i.e., $\vali \alloci -
\paymenti \geq \vali \alloci[j] - \paymenti[j]$.  The envy-free
benchmark (with budgets) is defined by optimizing over all envy-free
outcomes (that are budget respecting).  The following lemma
characterizes envy-free outcomes for valuation profiles $\vals$ that
are, without loss of generality, indexed by value, i.e., $\vali$'s are
monotonically non-decreasing in $i$.

%

\begin{lemma}[\citealp{HY-11}]
\label{l:envyfree}
Allocation $\allocs$ has prices for which it is envy free if and only
if it is swap monotone, i.e., $\alloci \geq \alloci[i+1]$.  The
minimum and maximum payments for which such an $\allocs$ is envy-free
are are
$\paymenti^{\min}=\sum_{j=i+1}^{n}(\alloci[j-1]-\alloci[j])\vali[j]$
and $\paymenti^{\max}=\sum_{j=i}^n(\alloci[j]- \alloci[j+1])\vali[j]$,
respectively.
\end{lemma}
Notice that envy-free payments are monotone non-decreasing in agent
values so an envy-free outcome is budget feasible if and only if the
highest valued agent (i.e., agent 1) has payment
$\paymenti[1] \leq\budget$.

The envy-free optimal benchmark for welfare and revenue with budgets
is defined by optimizing over all envy-free outcomes with respect to
the respective objective.  (These benchmarks are further characterized
in Section~\ref{s:benchmarks}.)
\begin{align*}
\EFOW(\vals,\budget) &= \max\big\{\sum\nolimits_i \vali\alloci: (\allocs,{\payments}) ~\hbox{\rm is EF, IR and BR} \big \}\\
\EFOR(\vals, \budget) &= \max\big \{\sum\nolimits_i \paymenti: (\allocs,{\payments}) ~\hbox{\rm is EF, IR and BR} \big \} 
\end{align*}
A prior-free guarantee about a mechanism's performance is defined as
follows.  A mechanism $\mech$ is a $\beta$-approximation to an
envy-free benchmark if its expected performance $\mech(\vals,\budget)$
is at least $\beta \EFO(\vals,\budget)$ for all $\vals$ and $\budget$.
For technical reasons, we slightly modify the envy-free benchmark for
revenue and instead approximate $\EFOR(\vals \super 2,\budget)$ where
$\vals \super 2 = (\vali[2],\vali[2],\vali[3],\ldots,\vali[n])$.  This
is necessary because, e.g., when $\vali[1] \gg n \vali[2]$, it is
impossible to approximate $\EFOR(\vali[1],\budget)$.  When the context
is clear, we will remove the superscripts and the budget and write
$\EFO(\vals)$ for readability.

%% file: benchmarks.tex
\section{The Envy-free Benchmark}
\label{s:benchmarks}

In this section we characterize welfare-optimal envy-free outcomes for
agents with a common budget in symmetric, cardinal environments; at
the end of the section we adapt the characterization to the objective
of revenue.  

Recall that in such an environment an allocation $\allocs =
(\alloci[1],\ldots,\alloci[n])$ is envy-free if and only if it is
swap-monotone and its minimum payments are given by the formula
$\paymenti^{\min} = \sum_{j=i+1}^{n} \vali[j] ( \alloci[j-1] -
\alloci[j])$ for all agents $i$ (Lemma~\ref{l:envyfree}).  Note that
in order to maximize welfare subject to a budget constraint, picking
the minimum envy-free payments is clearly optimal. Further, as
envy-free payments are monotone, it is sufficient to impose the budget
constraint only on the payment of the top agent, i.e.,
$\paymenti[1]$. Therefore, the welfare-optimal envy-free allocation
can be captured by the following linear program (LP).\footnote{Notice
  that the Lemma~\ref{l:envyfree} allows us replace the IC and IR
  constraint with a monotonicity constraint on $\allocs$ and it allows
  payment constraints to be expressed in terms of values.}
\begin{align}
\label{lp.efo}
\max\quad & \sum\nolimits_{i =1}^n \vali[i] \alloci\\
\notag
\textrm{s.t.}\quad &\alloci \geq \alloci[i+1] & \forall\, i\\
\notag
& \paymenti[1] = \sum\nolimits_{i=2}^{n} \vali[i] ( \alloci[i-1] - \alloci) \leq \budget.\\
\notag
& \allocs~ \textrm{is feasible}.\\
\intertext{The relaxation of this LP obtained by Lagrangifying the budget constraint is as follows.}
\label{lp.lefo}
\max\quad & \sum\nolimits_{i =1}^n \vali[i] \alloci - \lag \big( \sum\nolimits_{i= 2}^{n} \vali[i] ( \alloci[i-1] - \alloci)\big) + \lag \budget\\
\notag
\textrm{s.t.}\quad & \alloci \geq \alloci[i+1] &\forall\, i\\
\notag
& \allocs~ \textrm{is feasible}.
\end{align}
\begin{lemma}
\label{lem.lagrangian} 
An allocation is optimal for LP \eqref{lp.efo} if and only if, either 
\begin{itemize} 
\item for some choice of $\lag>0$, the allocation is optimal for the
Lagrangian relaxation \eqref{lp.lefo} with $\lag$,  and it satisfies
the budget constraint with equality, $\sum_{i=2}^{n}\vali[i]
(\alloci[i-1]-\alloci)=\budget$, or
\item the allocation is optimal for the Lagrangian relaxation
\eqref{lp.lefo} with $\lag=0$ and satisfies the budget constraint, 
$\sum_{i=2}^{n}\vali[i](\alloci[i-1]-\alloci)\leq\budget$.
\end{itemize}
\end{lemma}
\begin{proof} 
The statement of the theorem is equivalent to complementary slackness conditions characterizing optimal solutions of an LP. 
\end{proof}

%
%
We now consider the optimization problem given by the Lagrangian
relaxation \eqref{lp.lefo} for a fixed choice of $\lag$.  The
objective function of the Lagrangian relaxation \eqref{lp.lefo} can be
rewritten as $\sum_i \lvvi \alloci $ where
$\lvvi[1]=\vali[1]-\lag\vali[2]$ and $\lvvi=\vali
+\lag(\vali-\vali[i+1])$ for $i\geq 2$.  We refer to these as {\em
  Lagrangian virtual values}.  The Lagrangian relaxation is now simply
the problem of finding the Lagrangian virtual surplus optimal
allocation, subject to feasibility, swap-monotonicity, and, when
$\paymenti[1] = \budget$ when $\lambda > 0$.

%
%
Optimizing the Lagrangian virtual surplus $\sum_i \lvvi \alloci$ of
non-monotone virtual values subject to swap monotonicity (of the
allocation) can be simplified via the technique of {\em ironing}
\citep[cf.][]{M-81}.  The resulting ironed virtual values are monotone
and, therefore, ironed virtual surplus maximization without a
swap-monotonicity constraint on the allocation will always give an
allocation that is swap monotone.

Ironed Lagrangian virtual values are constructed as follows.  The {\em
  Lagrangian revenue curve} is the cumulative Lagrangian virtual value
$\lrev(j) = \sum\nolimits_{i=1}^j \lvvi = \sum\nolimits_{i=1}^j \vali
- \lag\vali[j+1]$.  The {\em ironed Lagrangian revenue curve} $\ilrev$
is the smallest concave function that is point-wise larger than $\lrev$
and the origin.  The {\em ironed Lagrangian virtual value} of $i$ is
the left slope of the ironed Lagrangian revenue curve, i.e., $\ilvvi =
\ilrev(i) - \ilrev(i-1)$.  An {\em ironed interval} $I =
\{i,\ldots,j\}$ is a sequence of consecutive agents where $\ilrev(i-1)
= \lrev(i-1)$, $\ilrev(j) = \lrev(j)$, and $\ilrev(k) > \lrev(k)$ for
$k \in \{i,\ldots,j-1\}$. The left-slope of the ironed Lagrangian
revenue curve for all $j \in I$ is the same; therefore, their ironed
Lagrangian virtual values are the same.

\begin{lemma}
\label{lem.ironing} 	
An allocation $\allocs$ is optimal for the Lagrangian relaxation
\eqref{lp.lefo} if and only if the allocation maximizes the ironed
Lagrangian virtual surplus and the allocation is constant over agents
in the same ironed interval.
\end{lemma} 
Lemma~\ref{lem.ironing} has the same proof as the corresponding lemma
of \citet{HY-11} for (non-Lagrangian) ironed virtual surplus
maximization.

We now describe a tie-breaking procedure for Lagrangian ironed virtual
surplus maximization that (a) serves agents in the same ironed
interval with the same probability (as per Lemma~\ref{lem.ironing})
and (b) meets the budget constraint with equality (as per
Lemma~\ref{lem.lagrangian}).  Notice ties may arise because agents
within the same ironed interval have same ironed virtual value,
because agents in consecutive ironed intervals may have the same
ironed virtual value, and because several sets of agents may have the
same cumulative ironed virtual value.  The first kind of tie must be
broken uniformly at random, the latter two kinds of ties must be broken so
as to meet the budget constraint with equality.

The tie-breaking rule we will give is based on agents' values.  Notice
that the objective of \eqref{lp.lefo} is the difference between the
social surplus and $\lambda \paymenti[1]$ (the payment of the top
agent scaled by $\lambda$).  Therefore, when there are ties in
Lagrangian virtual surplus, it must be that the tied allocation with
the maximum (resp.~minimum) surplus minimizes (resp.~maximizes) the
payment of the top agent.  This maximum payment must be over budget
and the minimum payment must be under budget.  Therefore, the
appropriate convex combination of these two allocations has payment
exactly equal to the budget.  The outcome produced is welfare optimal
for budgeted agents.

The following approach optimizes Lagrangian ironed virtual surplus
with tie-breaking to maximize or minimize the welfare subject to (a)
swap monotonicity and (b) agents within the same ironed interval
receiving the same probability of service.  To maximize welfare,
average the values of agents within each ironed interval and tie-break
to maximize this averaged welfare.  This ensures that the agents in
the same ironed interval are treated the same, but otherwise allows
the mechanism to optimize welfare over sets of agents with tied
Lagrangian ironed virtual surplus.  To minimize welfare we would like
to optimize the negative of the welfare over allocations with equal
Lagrangian ironed virtual surplus.  However, this could result in
failure of swap monotonicity as agents in successive ironed intervals
with the same Lagrangian ironed virtual value will be ranked in the
opposite order as required for swap monotonicity.  Therefore, to minimize
welfare, average the values of agents with equal Lagrangian ironed
virtual value (this includes the averaging of agents within the same
ironed interval, but additionally averages agents in successive ironed
intervals that have the same Lagrangian ironed virtual value), and
tie-break to minimize this averaged welfare.

For symmetric, ordinal environments we can be more precise about the
tie-breaking process above.  In particular, without budgets, the
welfare-optimal allocation is given by the greedy-by-value algorithm.
Any other swap monotone allocation can be thought of as starting with
the greedy outcome and then shifting some of the allocation from
higher valued agents to lower valued agents (by partially randomizing
their order).  Of course, any such reallocation lowers the welfare.
The allocation that maximizes Lagrangian ironed virtual surplus and
welfare (subject to Lemma~\ref{lem.ironing}) is the one that
randomizes the order of agents in each ironed interval and then
applies the greedy algorithm to this ordering.  The allocation that
maximizes Lagrangian ironed virtual surplus and minimizes welfare is
the one that randomizes the order of all sets of agents with equal
Lagrangian ironed virtual value.  As above, the appropriate convex
combination of these two allocations meets the budget constraint with
equality.

Notice that above we are taking the convex combination of a minimal
ironing and maximal ironing.  A set of consecutive ironed intervals
with the same Lagrangian ironed virtual value by this construction
will be {\em partially ironed}.  This partial ironing is absent in
existing characterizations of optimal mechanisms.  Partial ironing is
never necessary for revenue maximization without budgets (for which
the ironing technique was first developed) and prior work on welfare or
revenue maximization with budgets has restricted attention to
benevolent distributions where there is only a single ironed interval
(if any) that includes the highest-valued agent (and therefore there
is no partial ironing).

We summarize the discussion above with the following theorem.

\begin{theorem} \label{thm:EFO-char.} 
For symmetric, ordinal environments the optimal envy-free allocation
is given by Lagrangian ironed virtual surplus maximization (for an
appropriate Lagrangian variable) where agents within an ironed
interval are completely ironed and agents with equal Lagrangian ironed
virtual values are partially ironed (for an appropriate probability).
\end{theorem}

The characterization of the revenue-optimal envy-free outcome for
agents with a common budget is the same as above except for (a) the
specific formula for Lagrangian virtual values and (b) the
tie-breaking procedure.  The tie-breaking procedure is a bit more
complex than for the welfare objective.

Virtual values are derived starting from the maximum envy-free
payments $\paymenti^{\max} = \sum\nolimits_{j=i}^{n} \vali[j] (
\alloci[j] - \alloci[j+1])$ from Lemma~\ref{l:envyfree}.  The
objective revenue (without budgets) is given by maximization of the
virtual surplus for (non-Lagrangian) virtual values $i \vali -
(i-1)\vali[i-1]$ \citep{HY-11}.  Relaxing the budget constraint gives
Lagrangian virtual values $\lvvi[1]= \vali[1](1 -\lag)$ and $\lvvi =
(i-\lag) \vali - (i-1-\lag)\vali[i-1]$ for $i \geq 2$.  The Lagrangian
revenue curve is $\lrev(i) = (i-\lag) \vali$ with $\lrev(0) = 0$.

For tie breaking, notice that the analogous objective of the
Lagrangian relaxation \eqref{lp.lefo} is revenue minus $\lambda
\paymenti[1]$.  Therefore, among allocations with the same Lagrangian
ironed virtual surplus, the one with the highest revenue has the
highest payment of the top agent and the one with the lowest revenue
has the lowest payment of the top agent.  Revenue is, of course, equal to the (non-Lagrangian) virtual surplus.  Whereas for maximizing and
minimizing welfare the monotonicity of values implies that we should
either prefer to iron as little or as much as possible, for maximizing
revenue, the virtual values may not be monotone.  Therefore, ironing (i.e.,
averaging) can be good and bad for both maximizing and minimizing
revenue.  The following process averages the (non-Lagrangian) virtual
values correctly.  
Consider a set of consecutive Lagrangian ironed intervals with the
same Lagrangian ironed virtual value. Average the (non-Lagrangian)
virtual values within each interval, calculate the induced revenue
curve (by summing prefixes of these averaged virtual values), consider
the two-dimensional convex hull the point set that defined this
revenue curve. For maximum revenue, iron as for the upper convex hull;
for minimum revenue, iron as for the lower convex hull.  Optimizing
Lagrangian ironed virtual surplus with tie-breaking by averaged
virtual surplus (from the averaged virtual values calculated above)
gives the outcomes with the minimum and maximum payment of the top
agent.  Mixing between these appropriately to meet the budget
constraint with equality gives the revenue-optimal envy-free outcome.

%% file: clinching.tex
\section{Welfare approximation for agents with a common budget}
\label{s:welfare}

In this section we study the (polyhedral) clinching auction
of~\citet{GMP-12} in position environments with a common budget.  The
outcome of the clinching auction is fundamentally simpler in structure
than those of the optimal incentive-compatible auction and optimal
envy-free outcome.  A fundamental construct in incentive-compatible
and envy-free optimization is {\em ironing}, that is, randomizing
between agents whose values fall within a given interval.  In
Section~\ref{s:benchmarks} we characterized welfare-optimal envy-free
outcomes as having multiple disjoint ironed intervals.  Our first task
of this section is to give a similar discription of the outcome of the
clinching auction.  In these terms, the clinching auction has
(essentially) one ironed interval and it always contains the top
agent.  This ironed interval is partially ironed with the singleton
interval containing the next highest-valued agent.  We give a simple
closed-form expression for calculating exactly how this partial
ironing is performed.

The clinching auction is not welfare-optimal in two respects.  First,
given a Bayesian prior distribution, the clinching auction's expected
welfare is not generally optimal among all incentive compatible
mechanisms.  Second, though the outcome of the clinching auction is
envy free, it is not the welfare-optimal envy-free outcome.
Nonetheless, we show that the clinching welfare is a two-approximation
to the envy-free optimal welfare.

\subsection{The clinching auction for position environments}

\citet{GMP-12} generalize the clinching auction for budgeted agents
to ordinal environments.  In this section, we characterize the outcome
of this process for symmetric ordinal environments, a.k.a., position
environments.  At a high level, the clinching auction is described by
an ascending price-clock with agents clinching some of the supply at
the price as it increases.  By the symmetry of the environment and the
fact that values of agents above an offered price do not affect the
allocation, the budget and partial allocations are identical for each
agent that remains in the auction as the price increases.  The
clinching auction can thus be formulated as follows:

\begin{definition}[Clinching Auction] 
\label{d:clinching}
The {\em clinching auction} maintains an allocation and
price-clock that start from zero.  The price-clock ascends
continuously and the allocation and budget are adjusted as follows.
\begin{enumerate}
 \item Agents whose values are less than the price-clock are removed and
   their allocation is frozen.
 \item The {\em demand} of any remaining agent is the remaining budget
   divided by the price clock.
\item Each remaining agent {\em clinches} (and adds to their current
  allocation) an amount that corresponds to the largest fraction of
  their demand that can be satisfied when all other remaining agents
  are first given as much of their demand as possible
  (subject to the feasibility constraint).\footnote{This step is vague
    in general environments; however, in ordinal
    environments, i.e.,  where the greedy algorithm is optimal, it is precise.}
\item The budget and allocation are updated to reflect the amount
  clinched in the previous step.
\end{enumerate}
The auction ends when everyone is removed or the remaining budget is
zero.
\end{definition}

The reason that the clinching auction is relatively simple to
describe in position environments with a common budget is that the
feasibility constraint imposed by clinching auctions is one where
allocation probability of top positions can be shifted to bottom
positions (e.g., by randomizing), but not vice versa.  Therefore, an
allocation $\allocs$ (in decreasing order) is feasible for position
weights $\poss$ (in decreasing order) if the cumulative allocation
$\cumalloci = \sum_{j \leq i} \alloci[j]$ at each coordinate $i$ is at
most the cumulative position weight $\supplyi = \sum_{j \leq i}
\posi[j]$.

\begin{proposition} 
\label{prop:clinching-pareto-optimal}
The clinching auction is incentive compatible and Pareto optimal in
position environments with a common budget.
\end{proposition}

Pareto optimality means that there is no other reallocation of goods
and money that makes an agent strictly better off and no agents are
worse off.  Proposition~\ref{prop:clinching-pareto-optimal}, which is
a special case of a more general result of \citet{GMP-12} implies
implies the following structure on the outcome.  This structural
theorem generalizes one from \citet{DLN-08} (for single-item
auctions).  It shows that, essentially, the clinching auction is
ironing only the top agents.

\begin{theorem}
\label{t:char}
Order the agents and positions in decreasing order and let $\magent$ be
the highest-valued agent who pays strictly less than the budget.  Then,
\begin{enumerate}[(a)]
\item
the auction terminates the moment the price-clock exceeds
$\vali[\magent]$, 
\item agents with higher values than $\magent$ each receive the
same service probability (and pay the budget), 
\item agent $\magent$
receives at least the service probability of her corresponding position,
\item agents with lower values than $\magent$ each receive exactly
  the service probability of their corresponding positions, and
\item the outcome is envy free.
\end{enumerate}
\end{theorem}

\begin{proof}
An agent drops out of the clinching auction when her value is
exceeded; otherwise, the auction terminates with the price clock below
her value when the remaining budget is zero.  Let $\magent$ be the
last agent to drop out when her value is exceeded.  By the definition
of the clinching auction and symmetry, all higher-valued agents pay
the budget and receive the same probability of service.  Again by the
symmetry of the process there is no envy.

Now consider the agents $\magent, \ldots, n$ who are paying strictly
less than the budget.  Assume that initially all excess service probability
from the top $\magent-1$ is given to agent $\magent$.  Feasibility
implies that service probability cannot be shifted up from low-valued
agents to high-valued agents and Pareto optimality implies that
service probability cannot be shifted down.  Consider any probability
shifted down from a higher valued agent to a lower valued agent, as
these agents are not paying their budget, a Pareto improvement would
be for the higher valued agent to buy this shifted service probability
from the lower valued agent (at a per-unit price equal to her value).
Consequently, agents $\magent+1,\dots,n$ get their corresponding
position weight and agent $\magent$ gets at least her corresponding
position weight.

Finally, we show that the price clock stops immediately after it
exceeds $\vali[\magent]$.  Assume that the allocation probability of
$\magent$ is strictly higher than her corresponding position weight
(the case of equality is addressed by Theorem~\ref{t:charrefined}) and
suppose that the price clock continues to rise.  The actual values of
the agents who have not retired are never taken into account in the
behavior of the clinching auction.  Therefore we can lower
$\vali[\magent-1]$ to just below price clock and at which point
agent $\magent-1$ would retire (and not pay her full budget).  However, now
we have both $\magent$ and $\magent-1$ not paying their full budget
and $\magent$ is getting strictly more service probability than her
corresponding weight which contridicts the other results of this
theorem.
\end{proof}

\begin{figure}[th!]
\begin{center}
\begin{pspicture}(-1,-0.5)(11,5)
\pscustom[linestyle=solid,linecolor=gray,hatchcolor=red,fillstyle=vlines,hatchangle=45]{
 \psline(0,0.3)(1,0.3)(1,0.7)(2,0.7)(2,1.2)(2.4,1.2)(2.4,1.8)(3.2,1.8)(3.2,3.3)(0.0,3.3)(0,0.3)
 \moveto(9,0.5)
 \psline(9,0.5)(9,1)(10.5,1)(10.5,0.5)(9,0.5)
 \moveto(9,1.5)
 \psline(9,1.5)(9,2)(9.75,2)(9.75,1.5)(9,1.5)
 \moveto(9,2.5)
 \psline(9,2.5)(9,3)(9.5,3)(9.5,2.5)(9,2.5)}
\pscustom[linestyle=solid,linecolor=gray,hatchcolor=blue,fillstyle=vlines,hatchangle=135]{
 \pscurve(3.5,2.0)(4,2.7)(5,2.9)
 \psline(5,3)(5,3.5)(0,3.5)(0.0,3.3)(3.2,3.3)(3.2,2.0)(3.5,2.0)
 \moveto(9.75,1.5)
 \psline(9.75,1.5)(9.75,2)(10.5,2)(10.5,1.5)(9.75,1.5)
 \moveto(9.5,2.5)
 \psline(9.5,2.5)(9.5,3)(10,3)(10,2.5)(9.5,2.5)}
\pscustom[linestyle=solid,linecolor=gray,fillstyle=solid,fillcolor=yellow]{
 \psline(5,3.5)(5,3.7)(0,3.7)(0,3.5)(5,3.5)
 \moveto(3.5,2.0)
 \pscurve(3.5,2.0)(4,2.7)(5,2.9)
 \psline(5,2.0)(3.5,2.0)
 \moveto(10,2.5)
 \psline(10,2.5)(10,3)(10.5,3)(10.5,2.5)(10,2.5)}

\rput[l](10.7,0.75){$\budgeti[\magent]$}
\rput[l](10.7,1.75){$\budget$}
\rput[l](10.7,2.75){$\budgeti[\magent-1]$}
\psaxes[labels=none,ticks=none]{->}(9,4)
\psline(-0.1,3.7)(0.1,3.7)
  \psline[linestyle=dotted](0.1,3.7)(9.0,3.7)
  \rput[r](-.2,3.7){$\iposi[\magent-1]$}
\psline(-.1,3.3)(.1,3.3)
  \psline[linestyle=dotted](0.1,3.3)(9.0,3.3)
  \rput[r](-.2,3.3){$\iposi[\magent]$}
\psline(-.1,2.9)(.1,2.9)
  \psline[linestyle=dotted](.1,2.9)(5,2.9)
  \rput[r](-.2,2.9){$\alloci[\magent]$}
\psline(-0.1,2)(0.1,2)
  \psline[linestyle=dotted](0.1,2)(3.0,2)
  \rput[r](-0.2,2){$\posi[\magent]$}
  \rput[r](-.3,1.5){$\vdots$}
\psline(-0.1,0.7)(0.1,0.7)
  \psline[linestyle=dotted](0.1,0.7)(1.0,0.7)
  \rput[r](-0.2,0.7){$\posi[n-1]$}
\psline(-0.1,0.3)(0.1,0.3)\rput[r](-0.2,0.3){$\posi[n]$}
\psline(1,-0.1)(1,0.1)
  \rput[t](1,-0.2){$\vali[n]$}
\psline(2,-0.1)(2,0.1)
  \rput[t](2,-0.2){$\vali[n-1]$}
\rput[t](2.6,-0.2){$\cdots$}
\psline(3.2,-0.1)(3.2,0.1)
  \psline[linestyle=dotted](3.2,0.1)(3.2,1.8)
  \rput[t](3.2,-0.2){$\vali[\magent+1]$}
\psline(5,-0.1)(5,0.1)
  \psline[linestyle=dotted](5,0.1)(5,3)
  \rput[t](5,-0.2){$\vali[\magent]$}
\psline(6,-0.1)(6,0.1)
  \psline[linestyle=dotted](6,0.1)(6,3.7)
  \rput[t](6,-0.2){$\vali[\magent-1]$}
\rput[t](7,-0.2){$\cdots$}
\psline(8,-0.1)(8,0.1)
  \psline[linestyle=dotted](8,0.1)(8,3.7)
  \rput[t](8,-0.2){$\vali[1]$}
\psline[linecolor=red,linewidth=2pt](0,0.3)(1,0.3)(1,0.7)(2,0.7)(2,1.2)(2.4,1.2)(2.4,1.8)(3.2,1.8)(3.2,2.0)(3.5,2.0)
\pscurve[linecolor=red,linewidth=2pt](3.5,2.0)(4,2.7)(5,2.9)
\psline[linecolor=red,linewidth=2pt](5,2.9)(5,3.5)(9,3.5)
\psline[linestyle=none,showpoints=true,linecolor=red,linewidth=2pt](1,0.3)(2,0.7)(2.4,1.2)(3.2,1.8)(5,2.9)(6,3.5)(8,3.5)
\end{pspicture}
\caption{The outcome of the clinching auction is completely specified
  by this figure.  For agents $i < \magent$, the allocation rule
  $\alloci(z,\valsmi)$ of the clinching auction is depicted; the
  payment of these agents is equal to the area of shadded region which
  is equal to the budget $\budget$.  Agents $i \geq \magent$ share
  this allocation rule on $z \leq \vali$ which is the relevant portion
  of the allocation rule for calculating payments.  In the envy-free
  outcome that irones the top $i$ agents (and ignores the budget
  constraint) the payment of the top agent, denoted $\budgeti$, is
  depicted for $i \in \{\magent-1,\magent\}$.  It is clear from the
  picture these are monotone in $i$ and $\budgeti[\magent] < \budget
  \leq \budgeti[\magent-1]$.}
\label{f:clinchingcurve}
\end{center}
\end{figure}
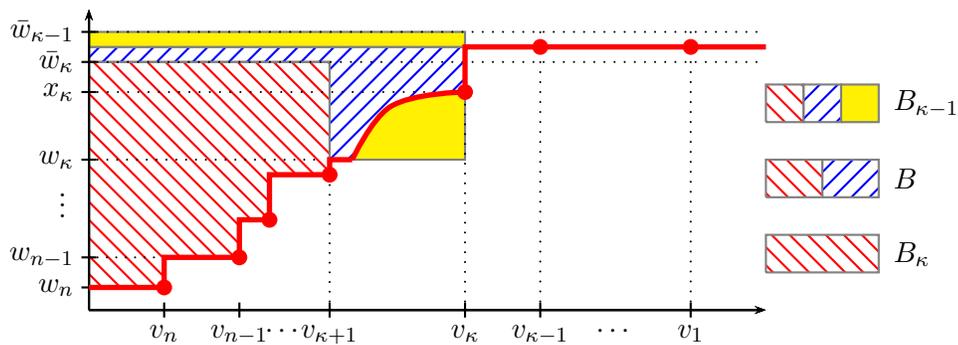

It is easy to see from Theorem~\ref{t:char} that the service
probability of each of the top $\magent-1$ agents is a little less
than the average weight of the top $\magent-1$ positions, and a little
more than the average weight of the top $\magent$ positions.  The
service probability of $\magent$ is between the weight of her
corresponding position and the average weight of the top $\magent$
positions.  In fact, the exact service probabilities (and
correspinding payments) can be precisely calculated.

The execution of the clinching auction can be described by two
phases.  In the first phase, the position weights and values are
binding; in the second phase, the budget is binding but the position
weights are not.  Consider ironing the top $i$ agents and the
associated minimum envy-free payments.  Agents $j \leq i$ are served
with probability $\iposi = (\posi[1] + \ldots + \posi)/i$ and their
payments are $\budgeti = \vali[i+1](\iposi - \posi) +
\sum_{j=i+1}^n\vali[j](\posi[j-1] - \posi[j])$.  The payment
$\budgeti$ is decreasing in $i$ (which is obvious as in Figure~\ref{f:clinchingcurve}) and is zero for $i=n$ and (if the
budget is binding) greater than the budget for $i=1$.  Let $\magent$
be such that $\budgeti[\magent] \geq \budget > \budgeti[\magent+1]$.
In the first phase each of the bottom $n-\magent$ agents will clinch
their corresponding positions.  Let $\budget'$ be the remaining budget
and let $\posi' = \posi - \posi[\magent-1]$ be the remaining weight of
position $i \leq \magent$ that has not been clinched (with average
weight $\iposi[\magent]' = \iposi[\magent] - \posi[\magent-1]$).  In
the next phase, the clinching auction will behave exactly like the
clinching auction for multi-unit envionments: the budget starts to
bind at a price clock at most $\vali[\magent]$ and then the instant
the price-clock exceeds $\vali[\magent]$ the remaining supply is
evenly clinched by the highest $\magent-1$ agents.
Figure~\ref{f:clinchingcurve} depicts the resulting outcome and
Theorem~\ref{t:charrefined} formalizes the observed structure.


\newcommand{\rbudget}{\budget'}
\newcommand{\rsupply}{\supply'}
\newcommand{\rsupplyi}[1][i]{\rsupply_{#1}}
\newcommand{\rcumfeas}{\cumfeas'}
\newcommand{\rcumfeasi}[1][i]{\rcumfeas_{#1}}
\newcommand{\rrbudget}{\rbudget'}
\newcommand{\rrsupply}{\rsupply'}
\newcommand{\rrsupplyi}[1][i]{\rrsupply_{#1}}
\newcommand{\rrrsupply}{\rrsupply'}
\newcommand{\rrrsupplyi}[1][i]{\rrrsupply_{#1}}
\newcommand{\rrcumfeas}{\rcumfeas'}
\newcommand{\rrcumfeasi}[1][i]{\rrcumfeas_{#1}}
\newcommand{\clinchamt}{\Delta}
\newcommand{\newclock}{\rclock}
\newcommand{\rclock}{\clock'}
\newcommand{\rrclock}{\clock''}
\newcommand{\startclock}{\clock}

\newtheorem*{theoremchar}{Theorem \ref{t:charrefined}}
\newcommand{\charrefined}{
For any position environment given by position weights $\poss$ and
budget $\budget$ satisfying $\budgeti[\magent] < \budget \leq
\budgeti[\magent-1]$ for some $\magent$, the polyhydral clinching
auction would allocate with:
\begin{enumerate}[(a)]
 \item $\posi$ to every $i\geq \magent+1$, and
 \item $\magent \iposi[\magent]$ split among
   the top $\magent$ agents evenly except for agent $\magent$ obtaining
   $\delta$ less,
\end{enumerate}
where $\marginal$ is a simple function of $\vali[\magent+1]$; $\vali[\magent]$; the remaining budget, denoted
$\rbudget$; and the unclinched supply, denoted $\rsupplyi[\magent]$, after agent $\magent + 1$ drops out.}

\begin{theorem}
\label{t:charrefined}
\charrefined
\end{theorem}

%% file: clinchingVsWelfare.tex
\subsection{Welfare approximation for ordinal environments}

We now show that the clinching auction which (essentially) irons only
the top positions, is a two-approximation to the envy-free optimal
welfare which may come from ironing an arbitrary number of consecutive
positions; moreover, this bound is tight.

\begin{theorem}
For any position environment with common budgets, the welfare obtained
by the clinching auction is a $2$-approximation to the envy-free
optimal welfare. Furthermore, this ratio is tight even for the
single-item environment.
\end{theorem}

\begin{figure}[!ht]
\begin{center}
\psset{unit=0.7cm} \subfigure[The upper bound \eqref{eq:budget-ub} is
  depicted pictorially. In the clinching auction, the payment of the
  highest valued agent (cross-hatched) is equal to the budget and at
  least the rectangle (striped) whose area is {$\vali[\magent]
    \iposi[\magent-1]$}.]{
\begin{pspicture}(-1,-0.5)(9,4.5)
\psline[linestyle=none,hatchcolor=gray,fillstyle=vlines,hatchangle=45](0,0)(5,0)(5,3.7)(0,3.7)
\psaxes[labels=none,ticks=none]{->}(9,4.5)
\psline(-.1,3.7)(.1,3.7)\psline[linestyle=dotted](.1,3.7)(8,3.7)\rput[r](-.2,3.7){$\iposi[\magent-1]$}
\psline(5,-0.1)(5,0.1)\psline[linestyle=dotted](5,0.1)(5,3.7)\rput[t](5,-0.2){$\vali[\magent]$}
\psline(6,-0.1)(6,0.1)\psline[linestyle=dotted](6,0.1)(6,3.7)\rput[t](6,-0.2){$\vali[\magent-1]$}
\rput[t](7,-0.2){$\cdots$}
\psline(8,-0.1)(8,0.1)\psline[linestyle=dotted](8,0.1)(8,3.7)\rput[t](8,-0.2){$\vali[1]$}
\psline[linecolor=red,linewidth=2pt](0,0.3)(1,0.3)(1,0.7)(2,0.7)(2,1.2)(2.4,1.2)(2.4,1.8)(3.2,1.8)(3.2,2.0)(3.5,2.0)
\pscurve[linecolor=red,linewidth=2pt](3.5,2.0)(4,2.7)(5,3)
\psline[linecolor=red,linewidth=2pt](5,3)(5,3.5)(9,3.5)
\psline[linestyle=none,showpoints=true,linecolor=red,linewidth=2pt](1,0.3)(2,0.7)(2.4,1.2)(3.2,1.8)(5,3)(6,3.5)(8,3.5)
\pscustom[linestyle=solid,linecolor=red,hatchcolor=red,fillstyle=vlines,hatchangle=135]{
\psline(0,0.3)(1,0.3)(1,0.7)(2,0.7)(2,1.2)(2.4,1.2)(2.4,1.8)(3.2,1.8)(3.2,2.0)(3.5,2.0)
\pscurve(3.5,2.0)(4,2.7)(5,3)
\psline(5,3)(5,3.5)(0,3.5)}
\end{pspicture}
}
$\qquad$
\subfigure[The lower bound \eqref{eq:budget-lb} is depicted
  pictorially.  The payment of the highest valued agent (striped) is
  equal to the budget and at least the rectangle (cross-hatched) whose
  area is {$\vali[\magent](\alloci[1]- \alloci[\magent])$}.]{
\begin{pspicture}(-1,-0.5)(9,4.5)
\psline[linestyle=none,hatchcolor=gray,fillstyle=vlines,hatchangle=45](0,2.4)(5,2.4)(5,3.9)(0,3.9)
\psaxes[labels=none,ticks=none]{->}(9,4.5)
\psline(-.1,3.9)(.1,3.9)\rput[r](-.2,3.9){$\alloci[1]$}
\psline(-.1,2.4)(.1,2.4)\rput[r](-.2,2.4){$\alloci[j]$}
\psline(5,-0.1)(5,0.1)\psline[linestyle=dotted](5,0.1)(5,3.7)\rput[t](5,-0.2){$\vali[\magent]$}
\psline(6,-0.1)(6,0.1)\psline[linestyle=dotted](6,0.1)(6,3.7)\rput[t](6,-0.2){$\vali[\magent-1]$}
\rput[t](7,-0.2){$\cdots$}
\psline(8,-0.1)(8,0.1)\psline[linestyle=dotted](8,0.1)(8,3.7)\rput[t](8,-0.2){$\vali[1]$}
\psline[linecolor=blue,linewidth=2pt](0,1.5)(2.0,1.5)(3,1.5)(3,2.4)(5,2.4)(5,3)(6,3)(6,3.9)(8,3.9)
\psline[linestyle=none,showpoints=true,linecolor=blue,linewidth=2pt](3,1.5)(5,2.4)(6,3)(8,3.9)
\psline[linestyle=solid,linecolor=blue,hatchcolor=blue,fillstyle=vlines,hatchangle=135](0,1.5)(2.0,1.5)(3,1.5)(3,2.4)(5,2.4)(5,3)(6,3)(6,3.9)(0,3.9)
\end{pspicture}
}
\caption{Proofs by picture of the upper and lower bounds on the budget $\budget$.}
\label{f:Bbounds}
\end{center}
\end{figure}
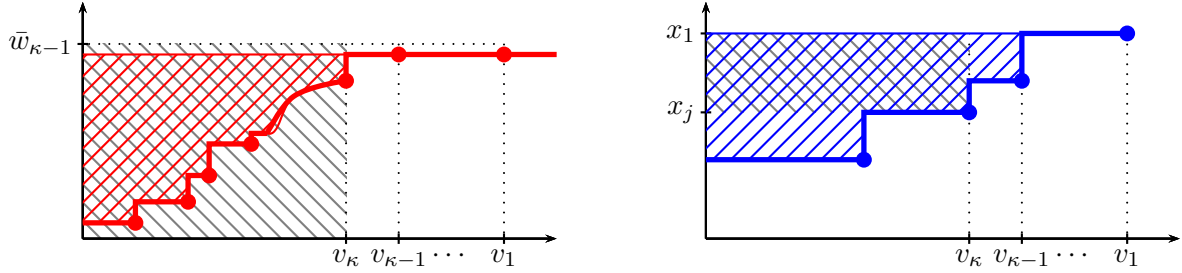

\begin{proof}
Let $\magent$ be the highest-valued agent who does not pay the budget
in the clinching auction.  Recall from Theorem~\ref{t:char} that,
relative to the outcome of the clinching auction, if we iron the top
$\magent$ agents (to get average service probability $\iposi[\magent]
= \sum_{i \leq \magent} \posi / \magent$) then agent $\magent$ gets
slightly more service probability at the expense of lowering the
service probability of the top $\magent-1$ agents; overall there is a
net decrease in welfare.
\begin{align}
\label{eq:clinching-lb}
\Clinching(\vals) \geq \sum_{i=1}^{\magent}\vali\iposi[\magent] +
      \sum_{i=\magent+1}^{n}
\vali\posi
\end{align}
Let $\allocs$ be the optimal envy-free allocation.  We know two things
about $\allocs$.  First, it is feasible, which means, in particular,
that $\sum_{i \leq \magent} \alloci \leq \magent \iposi[\magent]$,
i.e., the cumulative allocaiton is at most the cumulative supply.
Second, the payment of the highest-valued agent, i.e., $\paymenti[1]$,
(which is given by the ``area above the allocation rule'' as specified
by the mininimum envy-free payment identity of Lemma~\ref{l:envyfree})
is at most the budget.  We use these two bounds to show that
$\alloci[1] \leq 2\iposi[\magent]$.

The clinching auction ends when the price-clock just exceeds
$\vali[\magent]$, consequently the per-unit cost of service is bounded
by $\vali[\magent]$.  The probability of service clinched by the top
$\magent-1$ agents is slightly lower than $\iposi[\magent-1] =
\frac{1}{\magent-1} \sum_{i < \magent} \posi$.  Therefore an upper
bound on the maximum payment (and therefore the budget) is:
\begin{align}\label{eq:budget-ub}
\vali[\magent] \iposi[\magent-1] &\geq \budget.\\ 
  \intertext{In the
  envy-free optimal outcome the payment of the top agent
  (and therefore the budget) is at least:} 
\label{eq:budget-lb}
\budget 
  & = \sum_{i=2}^{n}(\alloci[i-1]-\alloci)\vali \geq \vali[\magent]
  (\alloci[1]-\alloci[\magent]).\\ 
    \intertext{The bounds \eqref{eq:budget-ub} and \eqref{eq:budget-lb} combine to
    give a bound on the probability of service $\alloci[1]$ of the
    top agent (and thus any agent) in the envy-free
    outcome.}
\label{eq:pos1a}
\alloci[1] &\leq \alloci[\magent] + \iposi[\magent-1].
\end{align}

The feasibility constraint of the position environment restricts the
envy-free outcome so that
\begin{align}\notag
\magent \iposi[\magent] & \geq \sum\nolimits_{i \leq \magent} \alloci 
\geq \alloci[1] + (\magent - 1) \alloci[\magent].\\
\intertext{Solving for $\alloci[1]$ we get a second upper bound.}
\label{eq:pos1b}
\alloci[1] & \leq \magent \iposi[\magent] - (\magent-1)\alloci[\magent].\\
\intertext{Add $(\magent-1)$ times \eqref{eq:pos1a} to \eqref{eq:pos1b} to
  get:} 
\magent \alloci[1] & \leq \magent \iposi[\magent] + (\magent-1)\iposi[\magent-1].
\end{align}
We conclude that $\alloci[1] \leq 2 \iposi[\magent]$ as desired.

For the optimal envy-free welfare problem, if the budget constraint is
replaced by the weaker constraint of $\alloci\leq 2\iposi[\magent]$,
the welfare can only get better. Furthermore, the optimal allocation
for this relaxed problem would shift as little service probability
down from top slots to lower slots as possible so as to meet the
allocation constraint that $\alloci \leq 2 \iposi[\magent]$.  As the
average weight of the top $\magent$ positions is $\iposi[\magent]$ the
probability of service for agent $\magent$ (which is the least of the
top agents) can only be at most the average.  Therefore, no additional
weight is shifted down to lower agents $j > \magent$ so,
\begin{align*}
\EFO(\vals) \leq \sum_{i=1}^{\magent} 2\vali\iposi[\magent] +
           \sum_{i=\magent+1}^n \vali\posi
      \leq 2\Clinching(\vals),
\end{align*}
where the last inequality follows from \eqref{eq:clinching-lb}.

To show that the 2-approximation is tight, consider the following
single-item scenario with a common budget of $\budget = 1$.  There are
$N+1$ agents; the highest valuation is $N^3$, the middle $N-1$
valuations are $N$, and the last valuation is $N-\epsilon$ where
$\epsilon$ is a small positive number.

The welfare-optimal envy-free allocation would serve the bottom $N$
agents with equal probability $\alloci[L]$ and the top agent with
probability $\alloc_H>\alloc_L$.  By optimizing the welfare
$N^3\alloc_H+N^2\alloc_L$ with the budget constraint
$N(\alloc_H-\alloc_L)\leq 1$, and the supply constraint
$\alloci[H]+N\alloci[L]\leq 1$, we have $\alloci[H]=\tfrac{2}{N+1}$
while $\alloci[L]=\tfrac{N-1}{N(N+1)}$. Thus the optimal envy-free
welfare for this case is $2N^2-N$.

The clinching auction would not let anybody clinch as long as the
price clock is below $N$ since there are $N+1$ agents who demand at
least $\tfrac{1}{N}$ while we only have $1$ item. However, as soon as
the price-clock reaches $N$, the bottom agent drops out, and we have
$N$ agents left who demand $\tfrac{1}{N}$ each. Thus, each of the top
$N$ agents would receive exactly $\tfrac{1}{N}$ and pay the budget.
Thus the welfare for clinching in this case is $N^2+N-1$.

In the limit as $N$ approaches $\infty$, the ratio between the two
welfares approaches $2$.
\end{proof}

%% file: revenue-pe.tex
\section{Revenue approximation for agents with a common budget}
\label{s:revenue}

The main approaches to prior-free auctions for digital goods
generalize to symmetric cardinal environments (without budgets).
\citet{HY-11} generalized the random sampling auction, and
\citet{HH-12} generalized the consensus estimate profit extraction
auction.  In this section, we generalize the random sampling profit
extraction auction of \citet{FGHK-02} for digital good environments
to symmetric cardinal environments with a common budget.  The random
sampling profit extraction auction splits the agents into a market
and a sample, estimates the optimal profit from the sample, and then
attempts to extract that profit from the market.

A profit extractor is a mechanism that is given some extra
information and, if that information is correct, is able to extract a
corresponding profit.  For symmetric cardinal environments,
\citet{HH-12} gave a profit extractor that is parameterized by an
estimated valuation profile and is able to extract profit of at
least the envy-free optimal revenue for the estimated valuation
profile when that estimate is a coordinate-wise lower bound on the
true valuation profile.\footnote{The \citet{HH-12} profit extractor
is described in Definition~\ref{def:PER} of
Appendix~\ref{s:nobudget}.}  Our profit extractor below is a
simplification of that of~\citet{HH-12} generalized to the case where
agents have budgets.

\begin{definition} 
\label{def:PEC}
The {\em clinching profit extractor}, $\PE{\evals}$, is parameterized
by non-increasing valuation profile $\evals$.  It calculates the
optimal envy-free outcome $\eallocs$ for $\evals$ and then runs the
clinching auction for position weights $\eallocs$ on the true
valuation profile $\vals$.
\end{definition}

Assume that $\vals$ and $\evals$ are in non-increasing order.  Define
$\vals$ as {\em one-ahead after index $\domind$} for $\evals$ if
$\domind$ is the lowest index for which all $i > \domind$ satisfy
$\vali[i+1] \geq \evali$.  When $\domind = 0$ define $\vals$ as
{\em one-ahead dominating} $\evals$, denoted $\valsmi[1] \geq
\evals$.  The following lemma shows that the clinching profit
extractor on $\vals$ is able to obtain the contribution to the
optimal envy-free revenue for $\evals$ from agents $\{\domind+1,\ldots,n\}$.
 
\begin{lemma}
\label{l:PEC}
If $\vals$ one-ahead dominates $\evals$ then the clinching profit
extractor revenue is at least the estimated envy-free optimal
revenue, i.e., $\PE{\evals}(\vals) \geq \EFO(\evals)$; moreover, if
$\vals$ is one-ahead after index $\domind$ for $\evals$ then the
contribution to the profit extractor revenue from agent $i >
\domind$ is at least the contribution to the estimated envy-free
optimal revenue from $i$, i.e., $\PE{\evals}_i(\vals) \geq
\EFO_i(\evals)$.
\end{lemma}

\begin{proof}
We will prove the second part of the lemma which implies the first.
Consider an $i > \domind$.  The maximum $i$ could pay in any outcome
is the budget, so if $i$ pays her budget in the clinching auction then
the bound holds.  Suppose instead that $i$ pays strictly less than her
budget in the clinching auction.  Consider the following sequence of
inequalities with explanation below (where $\allocs$ is the allocation
of the clinching auction and $\eallocs$ is the envy-free optimal
outcome for $\evals$).
\begin{align*}
\PE{\evals}_i(\vals)
  &\geq \sum_{j=i+1}^{n}(\alloci[j-1]-\alloci[j])\vali[j]\\ 
  &\geq \sum_{j=i+1}^{n}(\ealloci[j-1]-\ealloci[j])\vali[j]\\
  &\geq \sum_{j=i+1}^{n}(\ealloci[j-1]-\ealloci[j])\evali[j-1]\\
  &= \EFO_i(\evals).
\end{align*} 
The first inequality follows from envy freedom of the clinching auction
and the formula for minimum envy-free payments
(Lemma~\ref{l:envyfree}).  For $j > i$, Theorem~\ref{t:char} implies
that $\alloci[j] = \ealloci[j]$ because all but the highest-valued
agent who does not pay her budget are allocated with exactly their
corresponding position weight; the theorem also implies that $\alloci
\geq \ealloci[i]$ as agent $i$ also does not pay her budget (but she
might be the highest such agent).  The second equality then follows as
the service probabilities are unaffected by the swap from $\alloci[j]$
to $\ealloci[j]$ except for $\alloci$ which only appears positively and
is at least $\ealloci$.  The third inequality comes from the fact that
$i$ is greater than $\domind$ so one-ahead dominance holds at $i$ and
higher indices.  The final equality follows from the formula for
maximum envy-free payments (Lemma~\ref{l:envyfree}).
\end{proof}

%% file: revenue-sampling.tex
We now define a simple biased sampling procedure and show that it
ensures one ahead dominance with significant probability.

\begin{definition}[Biased Sampling]
Parameterized by a probability $\coin$, the {\em biased sampling} process
assigns each agent into the sample $\sample$ independently with probability
$\coin$, otherwise the market $\market$.
\end{definition}

Let $\mals$ and $\sals$ be the sorted valuation vectors of $\market$
and $\sample$ respectively, and assume that $\mals$ and $\sals$ are
padded with $0$'s to be equal in length for comparison convenience.
The biased sampling process has the following probabilistic properties
(proof given at the end of the section).

\begin{lemma}
\label{l:dominanceprob}
For the biased sampling process with $\coin<0.5$ and $\domind$ being a
random variable for the index after which $\mals$ is one-ahead for
$\sals$,
\begin{enumerate}[(a)]
\item $\prob{\mals\not\geq\sals}\leq\tfrac{\coin}{1-\coin}$,
\item\label{l:nondominanceprob}$\prob{\mals\not\geq\sals\given1\in\market}\leq
       \big(\tfrac{\coin}{1-\coin}\big)^2$, and
\item $\sum_{i=1}^{n}i\prob{\domi[i]\given1\in\market}\leq
       \tfrac{\coin}{(1-2\coin)^2}$.
\end{enumerate}
Furthermore, all inequalities are tight in the limit as $n$ approaches $\infty$.
\end{lemma}

\begin{lemma}
\label{l:randomselection}
The optimal envy-free revenue of a random sample $\sample$ whose
elements are selected i.i.d.~with probability $\coin$ satisfies
$\expect{\EFO(\sals)}\geq\coin\EFO(\vals)$.
\end{lemma}
\begin{proof}
Consider the envy-free optimal outcome for $\vals$.  Clearly if we
restrict attention only to the agents in $\sample$ there is still no
envy.  Therefore, $\EFO(\sals) \geq \EFO_\sample(\vals)$ where
$\EFO_\sample(\vals)$ is a short-hand notation for the contribution
from the agents in $\sample$ to the envy-free optimal revenue on
$\vals$.  Of course, $\expect{\EFO_\sample(\vals)}=\coin\EFO(\vals)$.
\end{proof}

%% file: revenue-approx.tex
\begin{definition}[$\BSPE{\coin}$]
\label{mech:BSPE}
 The \emph{biased sampling profit extraction} auction parameterized
 by $\coin < 0.5$ for a common budget $\budget$ works as follows.
 \begin{enumerate}
 \item Partition the set of agents into $\mals$ and $\sals$ using
  biased sampling parameterized by $\coin$.
 \item Run the clinching profit extractor parameterized by $\sals$ on
   $\mals$ and budget $\budget$.
\end{enumerate}
\end{definition}

Incentive compatibility of $\BSPE{\coin}$ comes straight from that
of the profit extractor. We have the following revenue guarantee.

\begin{lemma}
\label{l:approx:BSPE}
For all $p<0.5$ the revenue of $\BSPE{\coin}$ satisfies,
\begin{align*}
\BSPE{\coin}(\vals)\geq(1-\coin)\coin\EFO(\valsmi[1])-\tfrac{\coin(1-\coin)}{(1-2\coin)^2}\EFO(\vali[2]).
\end{align*}
\end{lemma}
\begin{proof}
Condition on the case that the highest-valued agent, i.e., $1$, is in
$\market$ and let $\domind$ is the index after which $\mals$ is
one-ahead for $\sals$.  From Lemma~\ref{l:PEC}, the profit extractor's
revenue conditioned on $\domi$ is
$\expect{\PE{\sals}(\mals)\given\domi}
\geq\expect{\EFO(\sals)\given\domi}-\sum_{j=1}^i\expect{\EFO_j(\sals)
  \given\domi}$.  This inequality holds since we would extract the
payment from all agents that are lower than $i$; or equivalently, we
would extract the full payment (the first term of the right-hand
side) minus those from the $i$ highest agents (the second term).
Using the observation that $\EFO_j(\sals)\leq \EFO(\vali[2])$ for all
$j$ (because agent $1$ is in $\market$), we have:
\begin{align*}
\expect{\PE{\sals}(\mals)\given\domi}
\geq\expect{\EFO(\sals)\given\domi}-i\EFO(\vali[2]).
\end{align*}
Summing these revenue guarantees over all $\domind$, we have:
\begin{align*}
\BSPE{\coin}(\vals)
&=\sum\nolimits_{i=1}^\infty\expect{\PE{\sals}(\mals)\given\domi}\prob{\domi}\\
&\geq\sum\nolimits_{i=1}^\infty\expect{\EFO(\sals)\given\domi}\prob{\domi}\\
&\qquad-\EFO(\vali[2])\sum\nolimits_{i=1}^\infty i\prob{\domi}\\
&=\expect{\EFO(\sals)}-\EFO(\vali[2])\tfrac{\coin}{(1-2\coin)^2}\\
&\geq\coin\EFO(\valsmi[1])-\EFO(\vali[2])\tfrac{\coin}{(1-2\coin)^2}.
\end{align*}
The last inequality comes from Lemma~\ref{l:randomselection} on
$\valsmi[1]$. Finally, we remove the conditioning on $1\in \market$ by
multiplying the above quantity by the probability $1-\coin$.  
\end{proof}

\begin{definition}[Pseudo-Vickery]
The {\em pseudo-Vickrey} auction finds the feasible outcome $\allocs$
that optimizes $\alloci[1]$ with $\alloci[j] = 0$ for $j\neq 1$ and
runs the clinching auction with position weights $\allocs$.
\end{definition}

\cite{HY-11} observe that $\EFO$ is subadditive (without budgets); it
continues to be subadditive with budgets (see
Lemma~\ref{l:subadditivity}).  Thus,
$\EFO(\valsmi[1])+\EFO(\vali[2])\geq\EFO\super{2}(\vals)$.\footnote{$\EFO\super{2} (\vals) = \EFO(\vals \super 2)$ where $\vals \super 2 = (\vali[2],\vali[2],\vali[3],\ldots,\vali[n])$.}
Furthermore, since pseudo-Vickery obtains at least
$\EFO(\vali[2])$, we have the following result.

\newcommand{\mcoin}{q}
\begin{theorem}
The convex combination of the pseudo-Vickery auction (with
probability $\tfrac{\mcoin}{1+\mcoin}$) and $\BSPE{\coin}$ (with probability
$\tfrac{1}{1+\mcoin}$), where $\mcoin=(1-\coin)\coin+\tfrac{\coin(1-\coin)}
{(1-2\coin)^2}$, approximates $\EFO\super{2}(\vals)$ within a factor
of $1+\tfrac{1}{(1-\coin)\coin}+\tfrac{1}{(1-2\coin)^2}$. This ratio
is minimized at $\budgetapprox$ when $\coin=0.211$.
\end{theorem}

\begin{proof}[\NOTSTOC{Proof }of Lemma~\ref{l:dominanceprob}]
Consider the following infinite random walk on a straight line:
starting from position $0$, with probability $\coin$, move backward
one step; otherwise, move forward one step.  The position of this
random walk describes precisely the difference between the number of
agents in $\market$ and $\sample$, where positive value means
$\market$ has more agents than $\sample$. We will show the results as
equalities by a ``probability of ruin'' analysis of an infinite random
walk; inequalities follow for random walks that terminate after a
finite number $n$ of steps.
\begin{enumerate}
\item The event $\mals\not\geq\sals$ happens when there exists a time
that $\market$ has fewer agents than $\sample$.  Let $\ruin$ be the
probability of ruin, i.e., the random walk eventually takes one step
backward from the initial position, we have $\ruin = \coin + (1 -
  \coin) \ruin ^ 2$. The first component is the probability of taking
one step backward in the first step, and the second component is the
probability of the first step being a forward step, then eventually
take two steps backward.  Solving this equation for $\ruin \in (0,1)$
gives $\ruin = \coin / (1 - \coin)$.

\item When we condition on $1 \in \market$, our initial position is
$1$, not $0$, and the probability of ruin is $\ruin^2$.

\item We will first derive $\prob{\domi \given 1 \in \market}$ for
$i \geq 1$.  Since $i$ is the lowest index after which $\mals$ is
one-ahead for $\sals$, we must have (a) an equal partition amongst the
top $2i$ agents, (b) $\vali[2i+1]$ is in $\market$, and (c) from this
point on, the number of agents assigned to $\market$ is never fewer
than that from $\sample$.  Thus, by conditioned on the highest value
agent already in $\market$, we
have:
\begin{align}
\notag
&\prob{\domi[i]\given1\in\market}\\
\notag
&\qquad= \tbinom{2i-1}{i}\coin^{i}(1-\coin)^{i-1} (1-\coin)
  \big(1-\tfrac{\coin}{1-\coin}\big)\\
\label{eq:iPr}
&\qquad= \tbinom{2i}{i}\big[\coin(1-\coin)\big]^{i}
   \tfrac{1-2\coin}{2(1-\coin)}.
\end{align}
The Taylor's series expansion of $\tfrac{1}{\sqrt{1-4z}}$ for any
$0<z<1/4$ gives us
\begin{align}
\notag
\sum_{i=0}^\infty\tbinom{2i}{i}z^i & = \tfrac{1}{\sqrt{1-4z}}.\\
\intertext{By differentiating both sides with respect to $z$,
then multiplying them with $z$, we have}
\notag
\sum_{i=1}^\infty i\tbinom{2i}{i}z^i & = \tfrac{2z}{(1-4z)\sqrt{1-4z}}.\\
\intertext{For $z=\coin(1-\coin)$, we have $\sqrt{1-4z}=1-2\coin$.
Hence, this equality translates to}
\label{eq:ibinom}
\sum_{i=1}^\infty i\tbinom{2i}{i}\big[\coin(1-\coin)\big]^i
   & =\tfrac{2\coin(1-\coin)}{(1-2\coin)^3}.
\end{align}
Putting these all together,
\begin{align*}
&\sum_{i=1}^\infty i\prob{\domi[i]\given1\in\market}\\
&\qquad=\sum_{i=1}^\infty i\tbinom{2i}{i}\big[\coin(1-\coin)\big]^{i}
   \tfrac{1-2\coin}{2(1-\coin)}&\mbox{from~\eqref{eq:iPr}}\\
&\qquad=\tfrac{2\coin(1-\coin)}{(1-2\coin)^3} \tfrac{1-2\coin}{2(1-\coin)}&\mbox{from~\eqref{eq:ibinom}}\\
&\qquad=\tfrac{\coin}{(1-2\coin)^2}.
\end{align*}
\end{enumerate}
\end{proof}

%% file: appendix-clinching.tex
\section{Closed-form characterization of the clinching auction.}
\label{s:closedform}

Here we give a detailed analysis of the clinching auction
(Definition~\ref{d:clinching}) in symmetric ordinal (a.k.a.~position)
environments.  Since position environments are symmetric and the
budgets are equal, agents who have not retired would clinch the same
amount and have the same budget. The feasibility constraints imposed
by the position environment dictate that $\cumalloci = \sum_{j=1}^i
\alloci[j] \leq \sum_{j=1}^i\posi[j] = \supplyi$, i.e., the cumulative
allocation is at most the cumulative supply.  Furthermore, at any
price-clock $\clock$ and budget $\budget$, each agent that can afford
$\clock$ would desire at most $\tfrac{\budget}{\clock}$ units. Thus,
the budget constraint means that $\cumalloci\leq
i\tfrac{\budget}{\clock}$, i.e., the cumulative allocation is at most
the cumulative demand.  These two constraints combine to restrict the
cumulative allocation by the minimum of the cumulative demand and
supply, i.e., $\cumalloci\leq\cumfeasi$ where
$\cumfeasi=\min\big(\supplyi,i\tfrac{\budget}{\clock}\big)$; we refer
to this as the {\em cumulative allocation constraint}.  The demand,
supply, and cumulative allocation constraints are all concave.

As agents simultaneously clinch some amount $\clinchamt$ at
price-clock $\clock$, the feasibility of the remaining allocation and
the budget would also change.  More specifically, $\supplyi$ would
become $\rsupplyi=\supplyi-i\clinchamt$, $\budget$ would become
$\rbudget=\budget-\clock\clinchamt$, and $\cumfeasi$ would become
$\rcumfeasi = \min\big(\rsupplyi,i\tfrac{\rbudget}{\clock}\big)$.  The
following lemma describes the clinching behavior of the agents as the
price-clock increases. Importantly, it shows that (a) when $i$ agents
remain, only the cumulative supply, demand, and allocation constraints
on $i-1$ and $i$ are relevant; and (b) after clinching the allocation
constraints on $i-1$ and $i$ are equal.  See Figure~\ref{f:clinchamt}
for a depiction of the relevant constraints in the clinching process.

\begin{lemma}
\label{l:clinchamt}
For a price clock $\clock\in(\vali[i+1],\vali]$ and cumulative
  allocation constraints $\cumfeasi[i-1]$ and $\cumfeasi$, the amount
  that each agent can clinch is precisely $\clinchamt =
  \cumfeasi-\cumfeasi[i-1]$, after which the newly induced cumulative
  allocation constraints are equal, i.e., $\rcumfeasi[i-1] =
  \rcumfeasi$, and the same (supply or demand) constraint is binding
  at $i-1$.
\end{lemma}

\begin{proof}
By definition of the clinching process, the amount that each agent can
clinch is that remaining when all other agents have first taken as
much as possible, this is exactly
$\clinchamt=\cumfeasi-\cumfeasi[i-1]$ when there are $i$ active
agents.  With an initial budget constraint of $\budget$ and price
clock of $\clock$, the budget and cumulative supply after clinching
are $\rbudget = \budget - \clock \clinchamt$ and $\rsupplyi = \supplyi
- i \clinchamt$.  Therefore the cumulative allocation constraint is
$\rcumfeasi=\min\big(\rsupplyi,i\tfrac{\rbudget}{\clock}\big)=\min\big(\supplyi-i\clinchamt,i\tfrac{\budget}{\clock}-i\clinchamt\big)=\cumfeasi-i\clinchamt$. Similarly,
$\rcumfeasi[i-1]=\cumfeasi[i-1]-(i-1)\clinchamt$. From the definition
of $\clinchamt$, then, the resulting cumulative allocation constraints
are equal, i.e., $\rcumfeasi[i-1] = \rcumfeasi$.
\end{proof}

\begin{figure}[ht!]
\begin{center}
\begin{pspicture}(-1,-0.5)(9,7)
\psaxes[labels=none,ticks=none]{->}(8,7)
\psline(-0.1,0.0)(0.1,0.0)\rput[r](-0.2,0.0){$0$}  
\psline(-0.1,1.4)(0.1,1.4)\rput[r](-0.2,1.4){$\cumfeasi[1]$}
\psline(-0.1,2.8)(0.1,2.8)\rput[r](-0.2,2.8){$\cumfeasi[2]$}
\psline(-0.1,5.4)(0.1,5.4)\psline[linestyle=dotted](0.1,5.4)(7.0,5.4)\rput[r](-0.2,5.5){$\cumfeasi$}
\psline(-0.1,5.1)(0.1,5.1)\psline[linestyle=dotted](0.1,5.1)(7.0,5.1)\rput[r](-0.2,5.0){$\cumfeasi[i-1]$}
\psline(1,-0.1)(1,0.1)\rput[t](1,-0.2){$1$}
\psline(2,-0.1)(2,0.1)\rput[t](2,-0.2){$2$}
\psline(6,-0.1)(6,0.1)\psline[linestyle=dotted](6,0.1)(6,5.1)\rput[t](6,-0.2){$i-1$}
\psline(7,-0.1)(7,0.1)\psline[linestyle=dotted](7,0.1)(7,5.4)\rput[t](7,-0.2){$i$}
\psline[linecolor=red,showpoints=true,linestyle=dashed](0,0)(1,2)(2,3)(3,3.8)(4,4.3)(5,4.7)(6,5.1)(7,5.4)
\psline[linecolor=blue,showpoints=true,linestyle=dashed]{->}(0,0)(1,1.4)(2,2.8)(3,4.2)(4,5.6)(5,7)
\psline[linecolor=purple,linewidth=2pt](0,0)(1,1.4)(2,2.8)(3,3.8)(4,4.3)(5,4.7)(6,5.1)(7,5.4)
\psline[linecolor=purple,showpoints=true](0,0)(1,1.1)(2,2.2)(3,2.9)(4,3.1)(5,3.2)(6,3.3)(7,3.3)
\rput[b](3.5,6){$\tfrac{\budget}{\clock}$}
\end{pspicture}
\caption{For a price-clock $\clock\in(\vali[i+1],\vali]$ the
cumulative supply constraints (red dashed line) can be represented by
a concave curve with the vertical axis representing the cumulative
allocation and the horizontal axis representing the number of agents
who can afford $\clock$.  The budget $\budget$ imposes a sumulative
demand constraint as given by a straight line with slope
$\tfrac{\budget}{\clock}$ (blue dashed line).  The cumulative
allocation constraint is the minimum of the cumulative supply and
demand constraints (thick purple line).  Any concave non-decreasing
curve that lies below the cumulative allocaton constraint is a
feasible cumulative allocation for the given supply and budget. After
the agents clinch some amount, the demand and supply constraints must
be adjusted (thin purple line) to take into account reduced budget and
remnant supply.}
\label{f:clinchamt}
\end{center}
\end{figure}
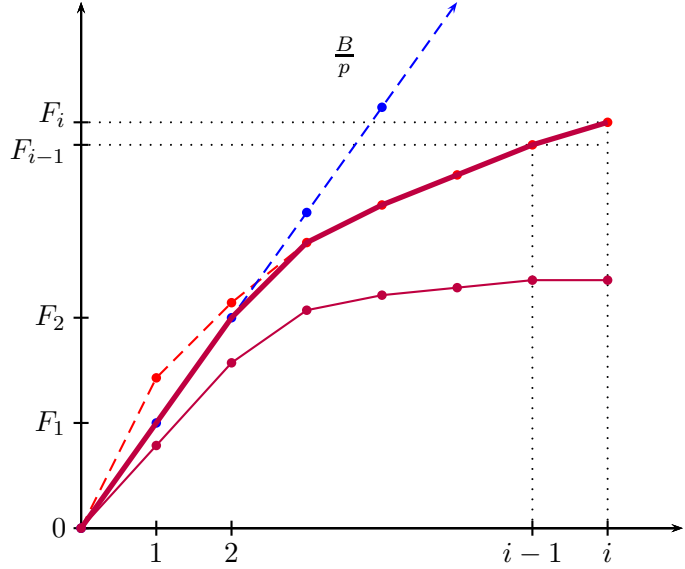

The clinching auction has two phases.  In the first phase, the supply
constraint is more restrictive than the demand constraint on $i-1$, i.e.,
$\supplyi[i-1] < (i-1) \frac{\budget}{\clock}$.  In the first phase,
nothing happens until either the price overtakes an agent's value at which point
the agent drops out and the number of active agents decreases from $i$ to
$i-1$ (and Lemma~\ref{l:clinchamt} describes how supply is clinched at
such a point) or the condition defining the first phase stops holding and
the second phase begins.  In the second phase, the two constraints are
equal, i.e., $\supplyi[i-1] = (i-1) \tfrac{\budget}{\clock}$.  Therefore,
the second phase begins with $\clock$ equal to $(i-1)
\tfrac{\budget}{\supplyi[i-1]}$ and, as the price increases, the demand
constraint binds allowing agents to gradually clinch.  The total amount
clinched in this second phase is described in the following lemma.

\begin{lemma}
\label{l:clinchprocess}
For cumulative supplies $\supplyi[i-1]$ and $\supplyi$, budget $\budget$,
and price clock $\startclock \in (\vali[i+1],\vali]$ such that
$\startclock = (i-1) \tfrac{\budget}{\supplyi[i-1]}$ (i.e., the supply and
  demand constraints are equal at $i-1$).  As the price clock increases to
$\vali$ each of the top $i$ agents clinch exactly $\tfrac{\supplyi}{i}
\big(1 - \big(\tfrac{\startclock}{\vali}\big)^{i}\big)$; the total remaining supply is $\supplyi \big(\tfrac{\startclock}{\vali}\big)^{i}$.
\end{lemma}

\newcommand{\funcf}[1]{G({#1})}
\newcommand{\funcfprime}[1]{g({#1})}

\begin{proof}
Let $\funcf{\newclock}$ be the total allocation that the active $i$
agents clinch while the price clock increases from $\startclock$ to
$\newclock$.  We will see that $\funcf{\cdot}$ is continuous and
differentiable; therefore, the total payment for this clinched amount is
given by integrating the price times amount clinched at each price $z
\in [\startclock,\newclock]$, i.e., $\int_{\startclock}^{\newclock} z
\funcfprime{z}\, dz$ (where $\funcfprime{z}$ denotes the derivative of
$\funcf{z}$ with respect to $z$).  The new cumulative
supply and budget are:
\begin{align*}
\rsupplyi &= \supplyi - i \funcf{\newclock}\\
\rbudget &= \budget - \int_{\startclock}^{\newclock} z
\funcfprime{z}\, dz.
\end{align*}
By Lemma~\ref{l:clinchamt}, $\rsupplyi = \rsupplyi[i-1] = (i-1)
\tfrac{\rbudget}{\newclock}$.  Therefore,
\begin{align*}
\supplyi - i \funcf{\newclock} &= \frac{(i-1)}{\newclock} \left(\budget - \int_{\startclock}^{\newclock} z \funcfprime{z}\, dz\right).\\
\intertext{Multiplying both sides by $\newclock$, differentiating with respect to $\newclock$, and simplifying; we obtain the differential equation,}
\supplyi - i \funcf{\newclock} &= \newclock \funcfprime{\newclock}.\\
\intertext{It is easy to check that the following is the solution to this differential equation that satisfies the initial condition $\funcf{\startclock} = 0$:}
\funcf{\newclock}&=\tfrac{\supplyi}{i}
\big(1-\big(\tfrac{\startclock}{\newclock}\big)^{i}\big).
\end{align*}
Plugging in $\newclock = \vali$ gives the desired result.
\end{proof}


We are now ready to derive the closed-form characterization of the
clinching auction in position environments. Recall from
Section~\ref{s:welfare} the definition of
$\budgeti = \vali[i+1](\iposi - \posi) +
\sum_{j=i+1}^n\vali[j](\posi[j-1] - \posi[j])$, the minimum
envy-free payment (ignoring the budget) when the top $i$ agents are
ironed (Figure~\ref{f:clinchingcurve}), and $\iposi = \supplyi / i$, the average weight of the top $i$ positions.

\begin{theoremchar}
\charrefined 
\end{theoremchar}

\begin{proof}
We give an alternative derivation of $\magent$ in terms of the
clinching process and then prove that this definition satisfies the
requirement of the theorem.  As described by Lemma~\ref{l:clinchamt},
when there are $i$ agents remaining, the amount clinched depends on
the cumulative allocation constraint on $i-1$ and $i$.  This
cumulative allocation constraint is the minimum of the supply and
demand constraints.  Let $\magent+1$ be the highest valued agent for
whom the demand constraint is strictly non-binding throughout the
clinching process.  The process proceeds as follows.
\begin{enumerate}
\item The price clock ascends from 0 to $\vali[\magent+1]$ and by the
  definition of $\magent$ the demand constraint (from the budget) does
  not bind during this process.  Lemma~\ref{l:clinchamt} implies that
  clinching only occurs as agents retire, i.e.,
  at prices equal to the agents' values.  At each price-clock $\vali$
  for $i\geq\magent+1$, agent $i$ drops out, while the top $i-1$
  agents would clinch $\posi[i-1] - \posi[i]$ with a payment of
  $\vali(\posi - \posi[i-1])$.

After this phase the price clock is $\rclock = \vali[\magent+1]$,
there are $\magent$ agents left, each agent $i\geq\magent$ is
allocated $\posi$, and each agent $i\leq \magent$ is allocated
$\posi[\magent]$.  The remaining budget for the top $\magent$ agents
and the cumulative supplies are:
\begin{align*}
\rbudget & = \budget - \sum_{i=\magent+1}^{n}\vali(\posi[i-1]-\posi)\\
\rsupplyi[\magent] = \rsupplyi[\magent-1] & = (\magent-1)(\iposi[\magent-1] - \posi[\magent]).
\end{align*}
\item\label{p:setup} With $\magent$ agents still active, the
  cumulative allocation constraint on $\magent-1$ agents is the
  minimum of cumulative supply $\rsupplyi[\magent-1]$ and the
  cumulative demand
  $(\magent-1)\tfrac{\rbudget}{\vali[\magent+1]}$. There are two
  cases depending on which is binding.
\begin{enumerate}
\item\label{p:nonbinding} If the demand constraint (from the budget)
  is not yet binding then the price clock will continue to ascend from
  $\vali[\magent+1]$ to $\rrclock = (\magent-1)
  \tfrac{\rbudget}{\rsupplyi[\magent-1]}$ when the demand constraint
  first (weakly) binds.  This must happen at $\rrclock \leq
  \vali[\magent]$ by the definition of $\magent$.  No additional
  supply is clinched and $\rrsupply = \rsupply$ and $\rrbudget =
  \rbudget$.  At this point the conditions of
  Lemma~\ref{l:clinchprocess} are satisfied.

\item\label{p:binding} If the demand constraint is (strictly)
  binding then, at the same price clock $\rrclock = \rclock = \vali[\magent + 1]$,
  additional supply is clinched.  Lemma~\ref{l:clinchamt} says that
  each of the top $\magent$ agents clinch an additional
  $\rsupplyi[\magent] - (\magent - 1) \tfrac{\rbudget}{\vali[\magent +
      1]}$.  The new budget and cumulative constraints are:
\begin{align*}
\rrbudget
  &= \magent \rbudget - \rsupplyi[\magent] \vali[\magent + 1],\\
\rrsupplyi[\magent] = \rrsupplyi[\magent-1]
  &= (\magent - 1)(\magent\rbudget/\vali[\magent + 1] - \rsupplyi[\magent]).
\end{align*}
At this point the conditions of
  Lemma~\ref{l:clinchprocess} are satisfied.

\end{enumerate}

\item \label{p:gradual} As the conditions of
  Lemma~\ref{l:clinchprocess} are met, the price clock now ascends from
  $\rrclock$ to $\vali[\magent]$ and an amount of
  $\tfrac{\rrsupplyi[\magent]}{\magent}
  \big(1-\big(\tfrac{\rrclock}{\vali[\magent]}\big)^{\magent}\big)$ is
  clinched by the top $\magent$ agents.  The remaining supply (by Lemma~\ref{l:clinchprocess}) is 
\begin{align*}
\rrrsupplyi[\magent] = \rrrsupplyi[\magent-1] &= \rrsupplyi[\magent]  \big(\tfrac{\rrclock}{\vali[\magent]}\big)^{\magent}.
\end{align*}



\item Agent $\magent$ retires as the price clock exceeds
  $\vali[\magent]$. As the linear demand constraint was binding at
  $\magent-1$ in Step~\ref{p:gradual}, only it remains.  Each of the
  $\magent -1$ active agents clinch an amount equal to the slope of
  this constraint and thereby equally split the entire remaining
  supply of $\rrrsupplyi[\magent-1]$.  The amount clinched is
\begin{align*}
\marginal = \tfrac{\rrrsupplyi[\magent-1]}{\magent-1} = \tfrac{\rrsupplyi[\magent]}{\magent -1}  \big(\tfrac{\rrclock}{\vali[\magent]}\big)^{\magent}
\end{align*}
where plugging in the appropriate $\rrclock$ and $\rrsupplyi[\magent]$ gives the
formula for $\marginal$ in terms of $\vali[\magent+1]$, $\vali[\magent]$, $\rbudget$,
and $\rsupplyi[\magent]$ as desired by the theorem.
\end{enumerate}

Finally, we argue that $\magent$ satisfies the definition from the
theorem statement, i.e., $\budgeti[\magent]<\budget \leq
\budgeti[\magent -1]$.  This follows easily from inspection of
Figure~\ref{f:clinchingcurve}.
\end{proof}

%% file: no-budget.tex
\section{Revenue approximation without budgets}
\label{s:nobudget}

In this section we give a variant of the biased sampling profit
extraction auction for bidders without budgets.  We are able to obtain
a better approximation ratio of $\nobudgetapprox$ (versus
$\budgetapprox$) when we use the profit extractor from \citet{HH-12}
that is incompatible with budgets. To get such an improvement, we
employed three techniques that cannot be utilized when agents are
budgeted, namely:
\begin{itemize}
 \item Always keep the highest value agent in the market.
 \item Reject (almost) everyone if there is no point-wise dominance.
 \item Always serve the highest agent, even when everyone else is rejected.
\end{itemize}

These will become clear in the definition of our mechanism. First, here
is the profit extractor that was introduced by~\citet{HH-12}.

\begin{definition}[$\PER{\evals}$]
\label{def:PER}
The \emph{profit extractor with rejection} parameterized by a
non-increasing valuation vector $\evals$ whose optimal envy-free
allocation is $\eallocs$, $\PER{\evals}$, would:
\begin{enumerate}
\item Sort the bids in a non-increasing order. If $\evali>\vali$ for
some $i$, reject everyone and charge nothing.
\item Serve agent $i$ with probability $\ealloci$ and charge the
appropriate IC payment.
\end{enumerate}
\end{definition}

\begin{lemma}[\citealp{HH-12}]
\label{l:PER}
For any $\vals\geq\evals$, the revenue of the Profit Extractor with
Rejection for $\evals$ on $\vals$ is at least the envy-free optimal
revenue for $\evals$.  Moreover, the inequality holds for each agent:
$\PER{\evals}_j(\vals)\geq \EFO_j(\evals)$.
\end{lemma}

This profit extractor can be used with a variant of the biased
sampling profit extraction mechanism that guarantees, in an incentive
compatible way, that the highest valued agent is in the market.  This
conditioning improves the probability that the market and sample
satisfy the dominance property necessary for the profit extractor to
succeed.

\begin{definition}[$\BSPE{\coin}$]
\label{mech:bspe}
The \emph{biased sampling profit extraction} auction parameterized by
$\coin<0.5$ works as follow.
\begin{enumerate}
\item\label{mech:bspe-pad} Pad the bids with a decreasing list of
infinitely many number of infinitesimally small positive values, e.g., 
$\vali[i]=\epsilon^i$, for $i>n$.
\item\label{mech:bspe-bs} Randomly assign each of the agents to one
of three groups $A$, $B$, and $C$ independently with probabilities
$\coin$, $\coin$, and $1-2\coin$, respectively.
\item\label{mech:bspe-v1} Assume without loss of generality that of
the highest valued agent in $A$ has value at least that of the
highest valued agent in $B$.  Define the market $\market = A \cup C$
and sample $\sample = B$.  (If this highest valued agent in $A$ wins
  in Step~\ref{mech:bspe-pe} and the second highest valued agent in
  $A \cup B$ is in $B$, increase her payment to this second highest
  value.)
\item\label{mech:bspe-pe} Run the profit extractor with rejection for
$\sals$ on $\mals$, where $\sals$ and $\mals$ are valuation vectors
of $\sample$ and $\market$ respectively.
\item\label{mech:bspe-vcg} If all agents are rejected by the profit
extractor and it is feasible to serve agent 1 (the highest valued
agent over all), serve her and charge her $\vali[2]$.
\end{enumerate}
\end{definition}

\begin{lemma}[Incentive Compatibility]
\label{l:IC}
For all probabilities $\coin < 0.5$, $\BSPE{\coin}$ is incentive compatible.
\end{lemma}

\begin{proof}
Fix the partitioning of $A$, $B$, and $C$.  No agent in $\market$ can
change the definition of sets $\market$ and $\sample$ without losing
(thus obtaining zero utility).  No agent in $\sample$ can change the
definition of sets $\market$ and $\sample$ without obtaining a
payment of at least her value (from the parenthetical in
  Step~\ref{mech:bspe-v1}, thus obtaining non-positive utility).
Therefore no agent wants to manipulate the definition of $\market$ and
$\sample$.  For given $\market$ and $\sample$, this mechanism is the
profit extraction mechanism which is incentive compatible for fixed
$\market$ and $\sample$.  Only the highest valued agent would want to
win in Step~\ref{mech:bspe-vcg}; furthermore, she cannot cause
dominance to fail without lowering her bid (and forfeiting her status
as the highest bidder).
\end{proof}

\begin{theorem}
\label{t:nobudget}
For any downward-closed permutation environment and any probability
$\coin<0.5$,\footnote{The first part of this lemma is non-trivial
  only for $\coin < 0.38$.} $\BSPE{\coin}$ approximates
$\EFO\super{2}(\vals)$ within factor of $\min\big\{\coin-\big(
  \tfrac{\coin}{1-\coin}\big)^2,\big(\tfrac{\coin}{1-\coin}\big)^2
  \big\}$.\footnote{$\EFO\super{2} (\vals) = \EFO(\vals \super 2)$ where $\vals \super 2 = (\vali[2],\vali[2],\vali[3],\ldots,\vali[n])$.} This ratio is optimized at $\coin=0.268$ which gives a
$\nobudgetapprox$ approximation.
\end{theorem}
\begin{proof}
Lemma~\ref{l:PER} says that we would obtain at least $\EFO(\sals)$
when $\mals\geq\sals$, while Step~\ref{mech:bspe-vcg} of the mechanism
guarantees $\EFO(\vali[2])$ otherwise. Furthermore,
Step~\ref{mech:bspe-pad} ensures the equality in part
(\ref{l:nondominanceprob}) of Lemma~\ref{l:dominanceprob}. Thus, the
expected revenue is at least:
\begin{align*}
\expect{\BSPE{\coin}(\vals)}\STOC{\\}
&\geq\expect{\EFO(\sals)\given\mals\geq\sals}\cdot\prob{\mals\geq\sals}\\
&\qquad + \EFO(\vali[2])\cdot\prob{\mals\not\geq\sals}\\
&=\expect{\EFO(\sals)} - \expect{\EFO(\sals)\given\mals\not\geq\sals}\cdot\prob{\mals\not\geq\sals}\\
&\qquad + \EFO(\vali[2])\cdot\prob{\mals\not\geq\sals}\\
&\geq\coin\EFO(\valsmi[1])-\EFO(\valsmi[1])\cdot\prob{\mals\not\geq\sals}\\
&\qquad + \EFO(\vali[2])\cdot\prob{\mals\not\geq\sals}\\
&=\big[\coin-\big(\tfrac{\coin}{1-\coin}\big)^2\big]\cdot\EFO(\valsmi[1])
     + \big(\tfrac{\coin}{1-\coin}\big)^2\cdot\EFO(\vali[2])\\
&\geq\min\big\{\coin-
  \big(\tfrac{\coin}{1-\coin}\big)^2,\big(\tfrac{\coin}{1-\coin}
  \big)^2\big\}\EFO\super{2}(\vals).
\end{align*}
The second inequality warrants some explanation: the first term
follows from applying Lemma~\ref{l:randomselection} to $\valsmi[1]=(\vali[2],\vali[3],\ldots)$,
the second term follows from monotonicity of $\EFO$. The last
inequality follows from the subadditivity of $\EFO$ (see
Lemma~\ref{l:subadditivity}, below) which implies that
$\EFO(\valsmi[1])+\EFO(\vali[2])\geq \EFO\super{2}(\vals)$.
\end{proof}

%% file: appendix-biasedsampling.tex
\section{Missing Proofs}
\label{app:proofs}

\citet{HY-11} gave the following proof that envy-free revenue is
subaditive; the proof is unaffected by the agents' budget.

\begin{lemma}
\label{l:subadditivity}
The envy-free optimal revenue for agents with a common budget is a
subadditive function.
\end{lemma}
\begin{proof}
Consider a partition of the original set of agents $N$ into $S$ and
$M$.  One way to obtain an envy-free outcome for $S$ (resp.~$M$) is to
calculate an envy-free outcome for the full set of agents $N$, and
then ignore the agents not in $S$ (resp.~$M$).  The optimal envy-free
outcome for $S$ (resp.~$M$) is certainly no wose.  Therefore, the
combined optimal envy-free revenue for $S$ and $M$ individually is at
least that of the full set $N$.
\end{proof}